\documentclass[11pt]{article}
\usepackage{aliascnt}
\usepackage{fullpage}

\usepackage{verbatim}
\usepackage{graphicx}
\usepackage{balance}  % for  \balance command ON LAST PAGE  (only there!)
\usepackage{subfig}  % use for side-by-side figures
\usepackage{xspace}
\usepackage[usenames,dvipsnames]{color}
\usepackage{epsfig} 
\usepackage{algorithm,algpseudocode}
\usepackage{hyperref}
\usepackage{amsmath, amssymb, amsthm}
\usepackage{proof}
\newtheorem{theorem}{Theorem}[section]
\newtheorem{lemma}[theorem]{Lemma}
\newtheorem{proposition}[theorem]{Proposition}
\newtheorem{corollary}[theorem]{Corollary}

\newenvironment{definition}[1][Definition]{\begin{trivlist}
\item[\hskip \labelsep {\bfseries #1}]}{\end{trivlist}}

\newcommand{\col}{\operatorname{Col}}

\newcommand{\addstepcomm}{\phi^{\mathrm {Array}}_{\mathrm{Comm}}}

\newcommand{\strip}{\phi_{\mathrm{Strip}}}

\newcommand{\cut}{\operatorname{Cut}}
 
\newcommand{\update}[1]{{#1}}

\begin{document} 
\pagestyle{plain}
\title{Towards Verifying Nonlinear Integer Arithmetic}
\author{Paul Beame\thanks{Research supported by NSF grants CCF-1524246 and SHF-1714593.}\\
{\small Computer Science and Engineering}\\
{\small University of Washington}\\
{\small beame@cs.washington.edu}
\and Vincent Liew$^*$\\
{\small Computer Science and Engineering}\\
{\small University of Washington}\\
{\small vliew@cs.washington.edu}}

\maketitle 

\begin{abstract}
%We give the first polynomial size resolution proofs for the commutative and distributive properties of a family of multiplier circuits. 
We eliminate a key roadblock to efficient verification of
nonlinear integer arithmetic using CDCL SAT solvers, by showing how to construct
short resolution proofs for many properties of the most widely used multiplier
circuits.
Such short proofs were conjectured not to exist.
More precisely, we give $n^{O(1)}$ size regular resolution proofs for
arbitrary degree 2 identities on array, diagonal, and Booth multipliers and
%quasipolynomial-
$n^{O(\log n)}$ size proofs for these identities on Wallace tree multipliers.
\end{abstract}

\section{Introduction}

The last few decades have seen remarkable advances in our ability to verify hardware and software.
Methods for hardware verification based on Ordered Binary Decision
Diagrams (OBDDs) developed in the 1980s for hardware equivalence
testing~\cite{bryant86} were extended in the 1990s to produce general methods for symbolic
model checking~\cite{bclmd94:modelchecking} to verify complex correctness properties of designs.  More recently,
several orders of magnitude of improvements in the efficiency of SAT
solvers have brought new vistas of verification of hardware and software
within reach.

Nonetheless, there is an important area of formal verification where
roadblocks that were identified in the 1980s still remain:
verification of data paths within designs for Arithmetic Logic Units (ALUs),
or indeed any verification problem in hardware or software that involves the
detailed properties of nonlinear arithmetic. 
Natural examples of such verification problems in software include computations
involving hashing or cryptographic constructions.
At the highest level of abstraction, nonlinear arithmetic over the integers is
undecidable, but the focus of these verification problems is on the
decidable case of integers of bounded size, which is naturally described in the
language of bit-vector arithmetic~(see, e.g. \cite{decision-procedures-book,DBLP:journals/mst/KovasznaiFB16}).

In particular, a notorious open problem is that of verifying properties of 
integer multipliers in a way
that both is general purpose and avoids exponential scaling in the bit-width.
Bryant~\cite{bryant91} showed that this is impossible using OBDDs since
they require exponential size in the bit-width just to represent the middle
bit of the output of a multiplier.   This lower bound has been 
improved~\cite{DBLP:journals/iandc/Bollig11} and extended
to include very tight exponential lower bounds
for much more general diagrams than OBDDs, including
FBDDs~\cite{pon95a,bw:readoncemult} and general
bounded-length branching programs~\cite{sw:multiplyRAM}.
With the flexibility of CNF formulas, efficient representation of multipliers is
no longer a problem but, even with the advent of greatly improved SAT solvers,
there has been no advance in verifying multipliers beyond exponential scaling.

One important technique for verifying software and hardware that
includes multiplication has been to use methods of uninterpreted functions
to handle multipliers~(see \cite{DBLP:conf/cav/BruttomessoCFGHNPS07,decision-procedures-book}) -- essentially converting them to black boxes and hoping
that there is no need to look inside to check the details.  Another important
technique has been to observe that it is often the case that one input to a
multiplier is a known constant and hence the resulting computation involves
linear, rather than nonlinear arithmetic.
These approaches have been combined with theories of arithmetic~(e.g. \cite{DBLP:conf/vlsid/BrinkmannD02,DBLP:conf/dac/ParthasarathyICW04,DBLP:conf/cav/BruttomessoCFGS08,DBLP:conf/tacas/BrummayerB09}), including preprocessors that do some form of rewriting to eliminate nonlinear arithmetic,
but these methods are not able, for example, 
to check the details of a multiplier implementation or handle 
nonlinearity.

Though the above approaches work in some contexts, they are very limited.
The approach of verifying code with multiplication using
uninterpreted functions is particularly problematic for 
hashing and cryptographic applications. For example, using uninterpreted
functions in the actual hash function computation inherently can never
consider the case that there is a hash collision, since it only can infer
equality between terms with identical arguments.  
Concern about the correctness of the arithmetic in such applications is real: for example, longstanding errors in 
multiplication in OpenSSL have recently come to light~\cite{openssl:montgomerybug}.

Recent presentations at verification conferences and workshops have highlighted
the problem of verifying nonlinear arithmetic, and multipliers in particular,
as one of the key gaps in our current verification methods~\cite{DBLP:conf/fmcad/Biere14,Biere:BIRS-TheoryApplicationsSAT-2014,DBLP:conf/fmcad/Kalla15,Biere:Fields-TheoryPracticeSAT-2016}.

Since bit-vector arithmetic is not itself a representation in Boolean variables,
in order to apply SAT solvers to verify
the designs, one must convert implementations and specifications to CNF
formulas based on specified bit-widths. 
The process by which one does this is called {\em flattening}~\cite{decision-procedures-book},
or more commonly {\em bit-blasting}.  
The resulting CNF formulas are then sent to the SAT solvers.   
%We focus on the core problem of verifying that a
%multiplier circuit correctly computes the product.
While the resulting bit-blasted CNF formulas for a
multiplier may grow quadratically with the bit-width, this growth is not
a significant problem.
On the other hand,  a major stumbling block for handling even modest bit-widths
is the fact that existing SAT solvers run on these formulas experience
exponential blow-up as the bit-width increases.  
This is true even for the best of recent methods, e.g., Boolector~\cite{DBLP:conf/tacas/BrummayerB09}, MathSAT~\cite{DBLP:conf/cav/BruttomessoCFGS08}, STP~\cite{DBLP:conf/cav/GaneshD07}, Z3~\cite{DBLP:conf/tacas/MouraB08}, and Yices~\cite{deMoura:Yices}.

In verifying a multiplier circuit one could try to compare it to a reference circuit that is known to be correct. 
This introduces a chicken-and-egg problem: how do we know that the reference circuit is correct?
Another approach to verifying a multiplier circuit is to check that it
satisfies the right properties. 
A correct multiplier circuit must obey the multiplication identities for a
commutative ring. 
If we check that each of these \emph{ring identities} 
holds then the multiplier cannot have an error.
This approach has the advantage that the specification of a multiplier
circuit can be written \emph{a priori} in terms of its natural properties,
rather than in terms of an external reference circuit.  

Empirically, however, modern SAT-solvers perform badly using either
approach to problems of multiplier verification. 
Biere, in the text accompanying benchmarks on the ring identities submitted to 
the 2016 SAT Competition~\cite{Biere-SAT-Competition-2016-benchmarks} 
writes that when given as CNF formulas, no known technique is capable of
handling bit-width larger than 16 for commutativity or associativity of
multiplication or bit-width 12 for distributivity of multiplication over
addition.
These observations lead to the question: is the difficulty inherent in these
verification problems, or are modern SAT-solvers just using the wrong
tools for the job?  

\begin{sloppypar}
Modern SAT-solvers are based on a paradigm called
conflict-directed clause-learning (CDCL)~\cite{mlm09:cdcl-sat-handbook} which can be seen as a way of breaking
out of the backtracking search of traditional DPLL solvers~\cite{DPLL62}. 
When these solvers confirm the validity of an identity (by not finding a
counterexample), their traces yield \emph{resolution} proofs~\cite{beamekautzsabharwal04} of that identity.
The size of such a proof is comparable to the running time of the solver;
hence finding short resolution proofs of these identities is a necessary
prerequisite for efficient verification via CDCL solvers.
Although it is not known whether CDCL solvers are capable of
efficiently simulating every resolution proof, all cases where short resolution
proofs are known have also been shown to have short CDCL-style traces~(e.g., \cite{DBLP:journals/lmcs/BussHJ08,bb2012:poolres,bk2014:stonepool}). 
\end{sloppypar}

The extreme lack of success of general purpose solvers (in particular CDCL
solvers) for verifying any non-trivial
properties of bit-vector multiplication, recently led Biere to
conjecture~\cite{Biere:Fields-TheoryPracticeSAT-2016} that
there is a fundamental proof-theoretic obstacle to succeeding on such problems;
namely, verifying ring identities for multiplication circuits, such as 
commutativity, requires resolution proofs that are exponential in the bit-width
$n$.

We show that such a roadblock to efficient verification of nonlinear arithmetic
does not exist by giving a general method for finding short resolution proofs
for verifying \emph{any} degree 2 identity for Boolean circuits
consisting of bit-vector adders and multipliers.
\update{This method is based on reducing the multiplier verification to
finding a resolution refutation of one of a number of narrow
\emph{critical strips}.}
We apply this method to a number of the most widely used multiplier circuits,
yielding $n^{O(1)}$ size proofs for array, diagonal, and Booth multipliers,
and $n^{O(\log n)}$ size proofs for Wallace tree multipliers.

These resolution proofs are of a special simple form: 
they are \emph{regular} resolution proofs\footnote{Some of these proofs are
even more restricted {\em ordered} resolution proofs, also known as {\em DP}
proofs, which are associated with the original Davis-Putnam
procedure~\cite{dp:prover}.  In contrast to the Davis-Putnam procedure,
which eliminates variables one-by-one keeping all possible resolvents, 
ordered resolution (or DP) proofs only keep some minimal subset of these
resolvents needed to derive a contradiction.}.
Regular resolution proofs have been identified
in theoretical models of CDCL solvers as one of the simplest kinds of proof that
CDCL solvers naturally express~\cite{DBLP:journals/lmcs/BussHJ08}. 
Indeed, experience to date has
been that the addition of some heuristics to CDCL suffices to find short 
regular resolution proofs that we know exist.   
The new regular resolution proofs that we produce are a key
step towards developing such heuristics for verifying general nonlinear
arithmetic.

\paragraph{Related work} 
\update{SAT solver-based techniques used in conjunction with case splitting
previously were shown to achieve some success for multiplier verification in
the work of Andrade et al.~\cite{DBLP:conf/ddecs/AndradeOFC07} improving
on earlier work~\cite{DBLP:conf/dac/AnderssonBCH02,DBLP:conf/date/RedaS01}
which combined SAT solver and OBDD-based ideas for multiplier verification
among other applications; however, there was no general understanding of when
such methods will succeed.}

\begin{sloppypar}
Recently, two alternative approaches to multiplier verification
have been considered: Kojevnikov~\cite{basolver} designed a mixed 
Boolean-algebraic solver, BASolver, that takes input CNF formulas in standard
format.  It uses algebraic rules on top of a DPLL solver. 
Though it can verify the equivalence of multipliers up to 32 bits in a 
reasonable time, in each instance it requires human input in order to find 
a suitable set of algebraic rules to help the solver.
An alternative approach using Groebner basis algorithms has been 
considered~\cite{DBLP:conf/date/Sayed-AhmedGKSD16}.   This is a purely 
algebraic approach based on polynomials.   Since the language of
polynomials allows one to explicitly write down the algebraic specification
for an $n$-bit multiplier, the verification problem is conveniently
that of checking that the multiplier circuit computes a polynomial equivalent
to the multiplier specification.
\cite{DBLP:conf/date/Sayed-AhmedGKSD16} shows that Groebner basis algorithms
can be used to verify 64-bit multipliers in less than ten minutes and 
128-bit multipliers in less than two hours.
One drawback of algebraic methods is they require
that the multipliers be identified and treated entirely
separately from the rest of the circuit or software. 
Unfortunately, for the non-algebraic parts of circuits,
Groebner basis methods can only handle problems several orders of magnitude
smaller than can be handled by CDCL SAT-solvers and it remains to be seen
whether it is possible to combine these to obtain effective 
verification for a general purpose software with nonlinear arithmetic
or circuits that contain a multiplier as just one component of
their design.
In contrast, CDCL SAT solvers are already very effective for the non-algebraic
aspects of circuits and are well-suited to handling the combination of different
components; our work shows that there is no inherent limitation preventing
them from being effective for verification of general purpose nonlinear
arithmetic.

Finally, independently of and in parallel with our results, there has also been further work on refining Groebner basis methods~\cite{bkr17}.   We postpone discussing that refinement until after we have presented our results.
\end{sloppypar}

\paragraph{Roadmap:} 

Section~\ref{sec:array} gives our polynomial size regular resolution proofs for array multipliers. Section~\ref{sec:diagonal} describes how to extend these ideas to obtain short proofs for diagonal and Booth multipliers. Section~\ref{sec:wallace} gives our quasipolynomial size regular resolution proofs for Wallace tree multipliers. 

\section{Notation and Preliminaries}
\label{sec:notation}

We represent Boolean variables in lowercase and
denote clauses by uppercase letters and think of them as sets of literals,
 for example $C=\{x,\bar{y},z\}$. 
 We will work with length $n$ \emph{bit-vectors} of variables, 
 denoted by ${\bf z} = z_{n-1}\ldots z_1 z_0$. When applicable, we will label arithmetic circuits by their output bitvector. For example, a multiplier with inputs $\mathbf{x,y}$ will be labeled $\mathbf{xy}$.

We consider identities from the commutative ring of integers $\mathbb{Z}$.
A variable assignment is denoted by a set $\sigma = \sigma(x_0, x_1 \ldots x_n) $ $= \{x_0 = b_0, x_1=b_1 \ldots x_n=b_n \}$,
where each $b_i \in \{0,1\}$.
$x_0,x_1,\ldots x_n$.
\begin{definition}
\label{def:ring}
A \emph{commutative ring} $(\mathcal{R},\oplus,\otimes,0,1)$ consists of a nonempty set 
$\mathcal{R}$ with addition ($\oplus$) and multiplication ($\otimes$) operators that
satisfy the following properties:
\begin{enumerate}
	\item ($\mathcal{R}$,$\oplus$) is associative and commutative and its identity element 
	is $0$.
	\item For each $\mathbf{x} \in \mathcal{R}$ there exists an \emph{additive inverse}.
	\item ($\mathcal{R}$,$\otimes$) is associative and commutative and its identity element 
	is $1 \not= 0$. 
	\item (distributivity) For all $\mathbf{x},\mathbf{y},\mathbf{z} \in \mathcal{R}$, 
	$\mathbf{x} \otimes (\mathbf{y} \oplus \mathbf{z}) = (\mathbf{x} \otimes \mathbf{y}) \oplus (\mathbf{x} \otimes \mathbf{z})$.
\end{enumerate}
A \emph{ring identity} $L = R$ denotes a pair of expressions $L,R$ that can be transformed into each other using commutativity, distributivity and associativity.
\end{definition}

\update{Note that both verifying integer $\oplus$ circuits and verifying that
$\mathbf{x} \otimes \mathbf{1} = \mathbf{x}$ are easy in practice, so
verifying an integer multiplier circuit $\otimes$ can be easily reduced to
verifying its distributivity.}

\begin{definition}
A {\em resolution proof} consists of a sequence of clauses, each of which is
either a clause of the input formula $\phi$, or follows from two prior
clauses via the {\em resolution rule} which produces clause $C\vee D$ from
clauses $C\vee x$ and $D\vee \overline{x}$. We say that this inference
\emph{resolves} the clauses \emph{on} $x$.
The proof is a {\em refutation} of $\phi$ if it ends with the empty
clause $\bot$.
(With resolution we will use the terms ``proof'' and ``refutation'' interchangeably, since resolution provides proofs of unsatisfiability.)
\end{definition}

We can naturally represent a resolution proof $P$ as a directed acyclic graph
(DAG) of fan-in 2, with $\bot$ labelling the lone sink node.
{\em Tree resolution} is the special subclass of resolution proofs where the
DAG is a directed tree.
Another restricted form of resolution is {\em regular resolution}:
A resolution refutation is \emph{regular} iff on any path in its DAG the
inferences resolve on each variable at most once.
The shortest tree resolution proofs are always regular.
An \emph{ordered} resolution refutation is a regular resolution refutation
 that has the further property that the order in which variables are resolved on
along each path is consistent with a single total order of all variables.
This is a very significant restriction and indeed the shortest tree resolution
proofs do not necessarily have this property.

We will find it convenient to express our regular resolution proofs in the form 
of a \emph{branching program} that solves the \emph{conflict clause search problem}.

\begin{definition}
Suppose that $\phi$ is an unsatisfiable formula. Then every assignment $\sigma$
 to its variables conflicts with some clause in $\phi$. 
 The \emph{conflict clause search} problem is to map any assignment to some corresponding conflicting clause.
\end{definition}

\begin{definition}
A branching program $B$ on the Boolean variables $X = \{x_0,x_1,\ldots\}$ 
and output set $\phi$ (typically a set of clauses in this paper) is a finite
 directed acyclic graph with a unique source node and sink nodes 
 at its leaves, each leaf labeled by an element from $\phi$. Each non-sink node is 
 labeled by a variable from $X$ and has two outgoing edges, one labeled $0$ and 
 the other labeled $1$. An assignment $\sigma$ \emph{activates} an 
 edge labeled $b \in \{0,1\}$ outgoing from a node labeled by the variable 
 $x_i$ if $\sigma$ contains the assignment $x_i = b$. 
 If $\sigma$ activates a path from the source to a sink 
 labeled $C \in \phi$, we say that the branching program $B$ \emph{outputs} $C$.

%We will also label each internal node of $B$ by the maximal assignment that is 
%consistent with reaching that node.

A \emph{read-once branching program} (also known as a Free Binary Decision 
Diagram, or FBDD) is a branching program where each variable is read at most 
once on any path from source to leaf. 
An \emph{Ordered Binary Decision Diagram (OBDD)} is a special case of an FBDD
in which the variables read along any path are consistent with a single total
order.
\end{definition}

The general case of the following proposition connecting regular resolution proofs
and conflict clause search is due to Krajicek~\cite{krajicek:book}; the 
special case connecting ordered resolution and OBDDs for the conflict clause
search problem was first observed in~\cite{lnnw95}.  We include its proof for completeness.

\begin{proposition}
Let $\phi$ be an unsatisfiable formula. A regular resolution refutation $R$ for 
$\phi$ of size $s$ corresponds to a size $s$ read-once branching program that 
solves the conflict clause search problem for $\phi$.

Suppose that $B$ is a read-once branching program of size $s$ solving the 
conflict clause search problem for $\phi$. 
Then there is a regular resolution refutation for $\phi$ of size $s$.

Furthermore, if $R$ is an ordered resolution refutation then the resulting 
branching program is an OBDD and if $B$ is an OBDD then the resulting resolution
refutation is an ordered resolution refutation.
\label{thm:prover_adv}
\end{proposition}

\begin{figure}[t]
  \centering 
\includegraphics[scale = 0.33]{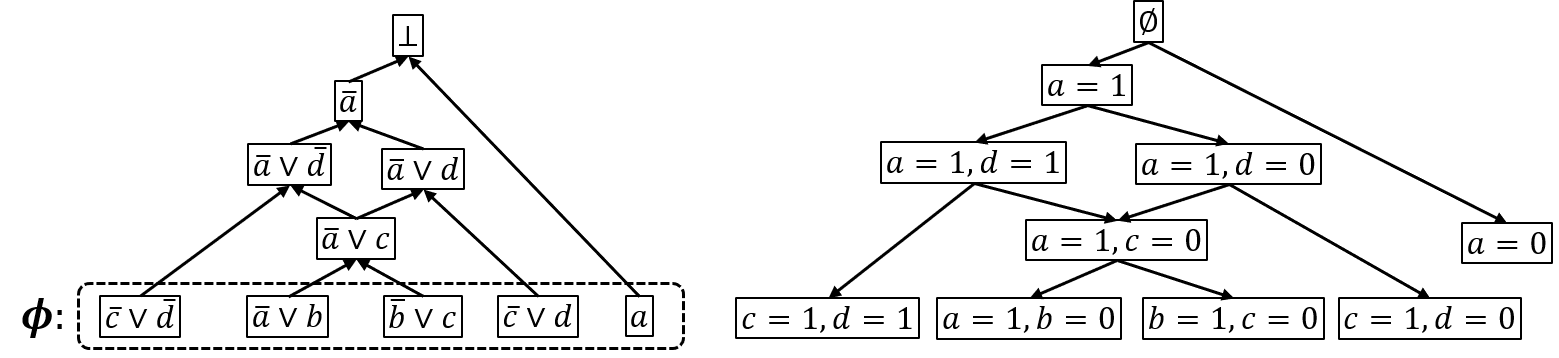}
\caption{A regular resolution refutation for $\phi$ and the corresponding branching program.}
\label{fig:resolutionBP}
\end{figure}

\update{
\begin{proof} 
Suppose that $R$ is a regular resolution refutation of size $s$ for $\phi$. 
Each clause $C$ appearing in $R$ is a node of $B$.
If two clauses $C_0 \vee x,C_1 \vee \bar{x}$ in $R$ resolve on a variable $x$ 
to produce the clause $C$, then in the branching program $B$ we branch from the 
node $C$ on the variable $x$ to reach $C_0 \vee x$ on the $x=0$ branch, and 
$C_1 \vee \bar{x}$ on the $x=1$ branch. The resulting branching program $B$ 
solves the conflict clause search problem for $\phi$ and has the same size 
as the refutation $R$.
The fact that no variable is branched on more than once on any path is immediate 
from the definition; the fact that this results in an OBDD in the case of
ordered resolution is also immediate.

In the other direction, we obtain a regular refutation $R$ from the specified 
read-once branching program $B$. We will label each node $v$ with the maximal clause $C_v$ that is falsified by
every assignment reaching $v$. These clauses form the regular resolution refutation.
If $v$ is a leaf then $C_v$ is the conflicting clause from $\phi$ found by $B$.
If $B$ branches from node $v$ on a variable $x$ to nodes $v_0,v_1$, then 
 in $R$ we resolve the clauses $C_{v_0},C_{v_1}$ on $x$ to obtain $C_v$. Again, the 
 number of clauses in the refutation $R$ is the same as the number of nodes 
 in the branching program $B$.
The fact that the resolution is regular follows immediately from the fact
that the branching program is read-once; if the branching program is an OBDD
then it is immediate that the resolution refutation is ordered.
\end{proof}}

In our proofs we represent each clause with the partial assignment it forbids. For example we write the clause $x \vee \bar{y}$ as the partial assignment $\{x = 0, y=1\}$.  
A branching program for conflict clause search in $\phi$ consists of three types of action, shown in Figures~\ref{fig:branch},~\ref{fig:propagate},~\ref{fig:merge}. At a node labeled by an assignment $\sigma \not\ni z$, we \emph{branch} on the variable $z$ by connecting a child node with assignment $\sigma \cup \{z=0\}$ using a $0$-labeled edge, and another child node $\sigma \cup \{z=1\}$, connected by a $1$-labeled edge. In the case that one of these children has an assignment conflicting with a clause $C \in \phi$, we say that we \emph{propagated} the assignment $\sigma$ to the other child's assignment. Lastly, for a set of leaf nodes with assignments $\sigma_0,\sigma_1,\ldots$ we can \emph{merge} their branches based on a common assignment $\sigma \subseteq \cap_i \sigma_i$ by replacing these nodes with a single node labeled by $\sigma$. 
%We also say, equivalently, that we merged based on \emph{forgetting} the assignment to $\cap_i \sigma_i \setminus \sigma$.

\begin{figure}[t]
\centering
\begin{minipage}{.45\textwidth}
\centering
\includegraphics[width=\linewidth]{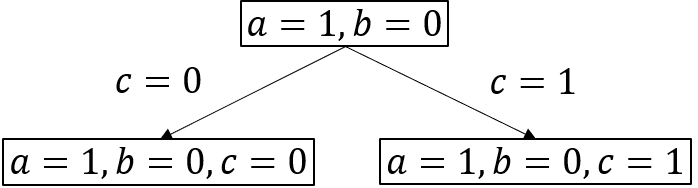}
\caption{Branching on $c$.}
\label{fig:branch}
\end{minipage}\hfill
\begin{minipage}{.45\textwidth}
\centering
\includegraphics[width=\linewidth]{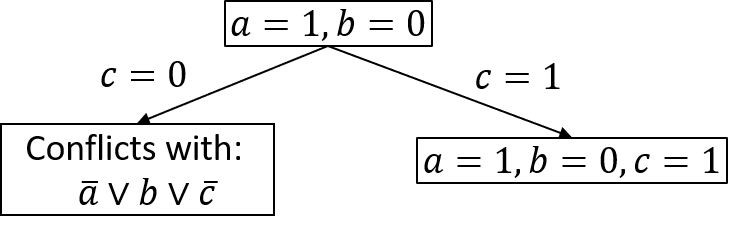}
\caption{Propagating to $c=1$.}
\label{fig:propagate}
\end{minipage}\hfill
\end{figure}
\begin{figure}[t]
  \centering 
\includegraphics[scale = 0.35]{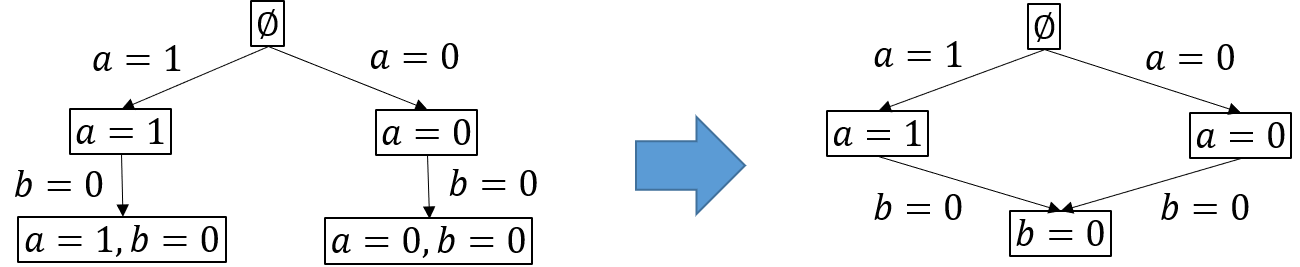}
\caption{Merging on the common assignment $\{b=0\}$.}
\label{fig:merge}
\end{figure}

\section{Array Multipliers}
\label{sec:array}
\subsection{Array Multiplier Construction}
We describe our SAT instances as a set of constraints, where each constraint is 
a set of clauses. Our circuits are built using $\emph{adders}$ that output, in 
binary, the sum of three input bits. An adder is encoded as follows:
\begin{definition}
Let $a_0,a_1,a_2$ be inputs to an adder $A$. The outputs $c,d$ of the adder 
$A$ are encoded by the constraints:
\[ d = a_0 \oplus a_1 \oplus a_2 \quad \quad c = MAJ(a_0, a_1, a_2) \] 
We call $c$ \emph{carry-bit} and $d$ the \emph{sum-bit}. 
If an adder has two constant $0$ inputs it acts as a \emph{wire}. If it has precisely one constant input 
$0$, we call it a \emph{half adder}. If no inputs are constant, we call it 
a \emph{full adder}.
\end{definition}

Each circuit variable has a \emph{weight} of the form $2^i$. 
Each adder will take in three bits of the same weight $2^i$ and output a 
\emph{sum-bit} of weight $2^{i}$ and a \emph{carry-bit} of weight $2^{i+1}$. 
The adder's definition ensures that the weighted sum of its input bits is the 
same as the weighted sum of its output bits. In the constructions that follow, 
we divide the adders up into columns so that the $i$-th column contains all 
the adders with inputs of weight $2^i$.

\paragraph{Ripple-Carry Adder:}
\begin{figure}[t]
  \centering 
\includegraphics[scale = 0.4]{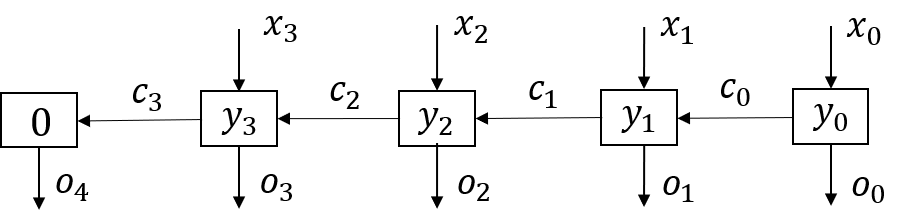}
\caption{4-bit ripple-carry adder adding $\mathbf{x,y}$. 
Each box represents a full adder with incoming arrows and outgoing arrows representing inputs 
and outputs.}
\label{fig:RCA}
\end{figure}
A ripple-carry adder, shown in Figure~\ref{fig:RCA},
 takes in two bitvectors $\mathbf{x,y}$ and outputs their sum in binary. 
In the $i$-th column, for $i \leq n$, we place an adder $A_i$ that takes 
the three variables $c_{i-1},x_i,y_i$ and outputs the adder's carry 
variable and sum variable to $c_i$ and $o_i$ respectively.
In the $(n+1)$-st column we place a wire $A_{n+1}$ taking $c_n$ as input and 
outputting to $o_{n+1}$. While the implementation is simple, it has depth $n$.

All the multipliers we describe perform two phases of computation to compute 
$\bf{xy}$. The first phase is the 
same in each multiplier: the circuit computes a \emph{tableau} of values 
$x_i \wedge y_j$ for each pair of input bits $x_i$ and $y_j$. These multipliers differ in the second 
phase, where the circuit computes the weighted sum of the bits in the tableau.

\paragraph{Array Multiplier:}
\label{sec:addstep}
An $n$-bit array multiplier works by arranging $n$ ripple-carry adders in order to sum the $n$ rows of the tableau. This multiplier has a 
simple grid-like architecture that is compact and easy to lay out physically. 
It has depth linear in its bitwidth.
In the first phase, an array multiplier computes 
each tableau variable $t_{ij} = x_i \wedge y_j$, with associated weight $2^{i+j}$.

\begin{figure}[t]
\centering
\begin{minipage}{.48\textwidth}
\centering
\includegraphics[width=\linewidth]{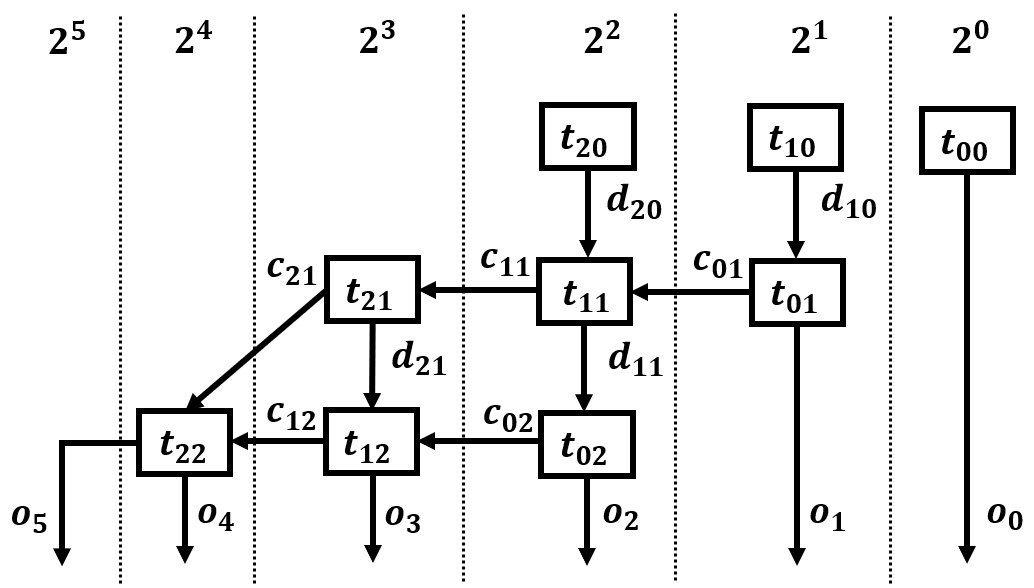}
\caption{3-bit array multiplier.}
\label{fig:addstep}
\end{minipage}\hfill
\begin{minipage}{.48\textwidth}
\centering
\includegraphics[width=\linewidth]{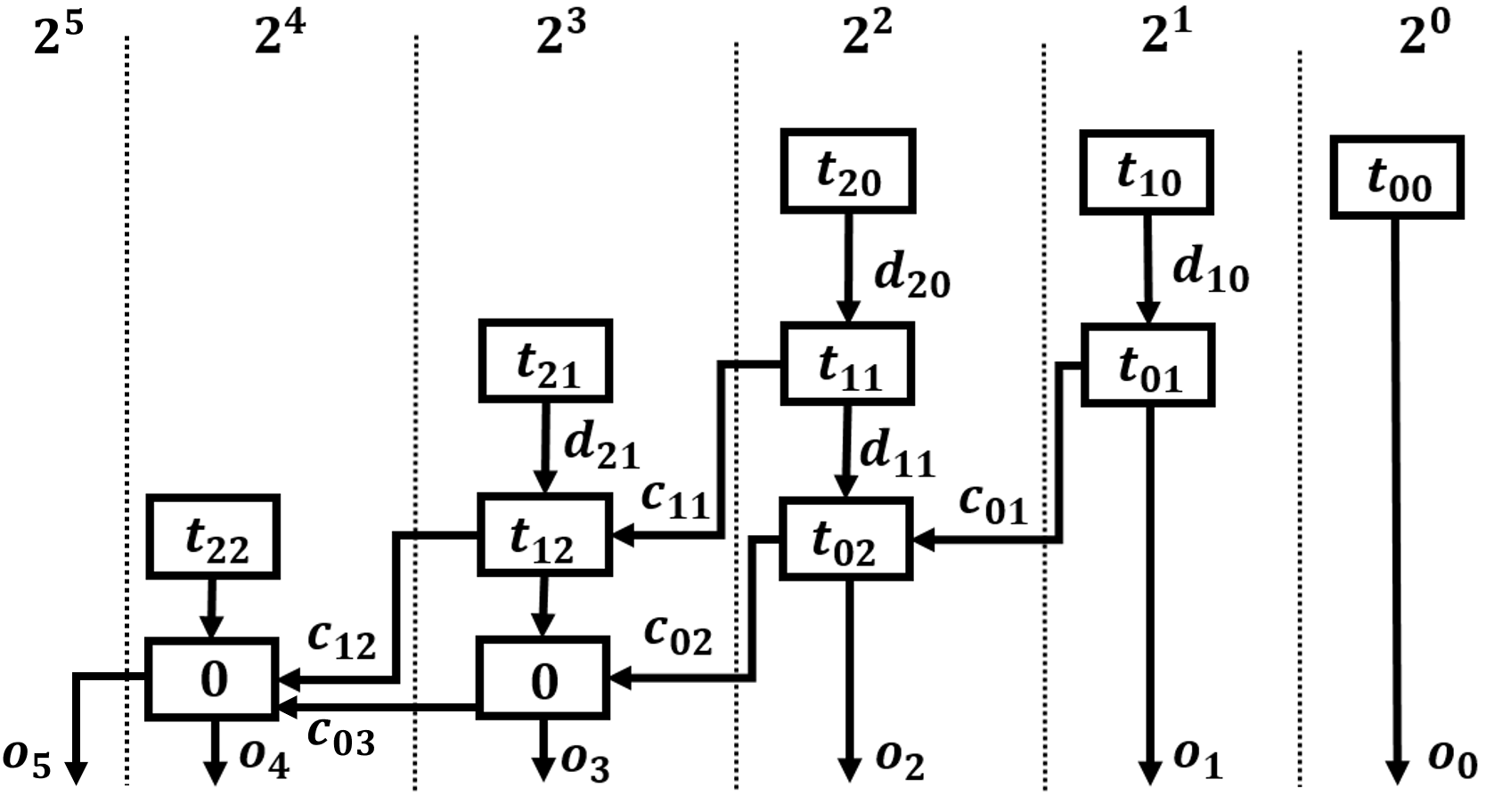}
\caption{3-bit diagonal multiplier.}
\label{fig:diagonal}
\end{minipage}\hfill
\end{figure}

Arrange a grid of full adders $A_{i,j}$, where $i,j \in [0,n]$, as shown in 
Figure~\ref{fig:addstep}. Adder $A_{i,j}$ occupies the $j$-th row and the $(i+j)$-th column and outputs the carry and sum bits $c_{i,j}$ and $d_{i,j}$. For $i<0$, adder $A_{i,j}$ takes inputs $t_{i,j}, d_{i+i,j-1}, c_{i-1,j}$ (replacing nonexistent variables with the constant $0$). Adders of the form $A_{n,j}$ take input $c_{n,j-1}$ instead of $c_{n-1,j}$.
Finally, we add constraints equating the sum-bits
$d_{0,0},d_{0,1}, \ldots,d_{0,n-1},d_{1,n-1},\ldots,d_{n-1,n-1}$ with the corresponding output bits $o_0,o_1,\ldots,o_{2n-1}$.

\subsection{Overview: Efficient Proofs for Degree Two Array Multiplier Identities}
\label{sec:addstep_identities}

We give polynomial-size resolution proofs that commutativity, distributivity, and the identity $x(x+1) = x^2 + x$ hold for a correctly implemented array multiplier. We go on to give polynomial-size resolution proofs for general degree two identities.

\paragraph{Proof Overview:}
The main idea, common to our proofs for each circuit family including Wallace tree multipliers, is to start by branching according to the lowest order disagreeing output bit between the two circuits. In each of these branches, the subcircuit, which we call a \emph{critical strip}, consists of the constraints on a small number of columns behind the disagreeing bit. For a large enough choice of width this critical strip is unsatisfiable since the removed section of the tableau on the right does not have enough total weight to cause the disagreeing output bit. It then remains to refute each critical strip.  

Our proofs inside each critical strip repeat three steps: (1) Branch on some of the input bits. (2) Propagate those values as far in the circuit as possible. (3) Save the resulting assignment to the boundary of the propagation. We call each of these boundaries a \emph{cut} in the circuit. 

These \emph{cuts} are sets of variables that, under any assignment, split the strip into a satisfiable and an unsatisfiable region. If a cut assignment was propagated from an earlier portion of the circuit, then this cut assignment is consistent with an assignment to this earlier subcircuit. But since the critical strip as a whole is unsatisfiable, this cut assignment must be inconsistent with any assignment to the rest of the circuit. Using these cuts, we reduce the unsatisfiable region in the critical strip until it is trivially refuted.

\update{One can view our proof as showing that the constraints within each
strip form a graph of \emph{pathwidth} $O(\log n)$ which, by~\cite{DBLP:conf/uai/Dechter96}, implies that there is a polynomial-size ordered resolution
refutation of the strip.  In the case of commutativity, our argument implies
that the constraint graphs
for the strips can be combined to yield a single constraint graph of
pathwidth $O(\log n)$.  For the other identities, the orderings on the strips
are different and the resulting constraint graphs only have small \emph{branchwidth} which, by~\cite{DBLP:conf/focs/AlekhnovichR02}, still implies that there
are small regular resolution proofs of the other identities.  Rather than simply
invoke these general arguments, we give the details of the resolution proofs,
along with more precise size bounds.}

\subsection{Proofs of Array Multiplier Commutativity}

\begin{definition}
We define a SAT instance $\addstepcomm(n)$. The inputs are length $n$ bitvectors $\mathbf{x,y}$. Using the construction from Section~\ref{sec:addstep}, we define array multipliers $L^{xy}$ and $R^{yx}$. The tableau variables are defined by the constraints 
\[t^{xy}_{i,j} = x_i \wedge y_j, \quad \quad t^{yx}_{i,j} = y_i \wedge x_j, \]
and in particular we can infer, through resolution, that $t^{xy}_{i,j} = t^{yx}_{j,i}$.

After specifying the subcircuits $L^{xy}$ and $R^{yx}$, we add a final subcircuit $E$, a set of \emph{inequality-constraints} encoding that the two circuits disagree on some output bit:
\[e_i = \big[ o^{xy}_i \neq o^{yx}_i \big] 
 \quad \forall i \in [0,2n-1],\]
 \[e_0 \vee e_1 \vee \ldots e_{2n-1}.\]
\end{definition}

We give a small resolution proof for $\addstepcomm(n)$ in the form of a labeled OBDD $B$, as described in Proposition~\ref{thm:prover_adv}. The variable order for $B$ begins with $e_0,e_1,\ldots$, followed by the output bits $o^{yx}_0, o^{yx}_1,\ldots$. Then $B$ reads the variables associated with adders $A^{xy}_{i,j}, A^{yx}_{j,i}$ in order of increasing $j$, reading each row right to left. Finally, $B$ reads the output bits $o^{xy}_0, o^{xy}_1,\ldots $, then the input bits $\mathbf{x},\mathbf{y}$ in an arbitrary order.

At the root of $B$, we search for the first output bit on which $L^{xy}$ and $R^{yx}$ disagree by branching on the sequences of bits $e_k=1, e_{k-1}=0, \ldots e_{0} = 0$ for each $k \in [0,2n]$. We will show that on each branch we can prove that $\addstepcomm(n)$ is unsatisfiable using only the constraints from $L^{xy}$ and $R^{yx}$ on the variables inside columns $[k-\log n,k]$.

\begin{definition}
Let $\Delta = \log n$.
Let $\strip(k) $ hold the constraints from $\addstepcomm(n)$ containing any tableau variable $t^{xy}_{i,j}$ or $t^{yx}_{i,j}$ for $i+j \in [k-\Delta,k]$.
Then add unit clauses to $\strip(k)$ to encode the assignment: $e_0=0, e_1=0, \ldots,e_{k-1}=0, e_k=1$. We call $\strip(k)$ a \emph{critical strip} of $\addstepcomm(n)$. We call the subset $\strip(k) \cap L$ the \emph{critical strip} of circuit $L$ and likewise for circuit $R$.
\end{definition}

\begin{lemma}
$\phi_{\mathrm{Strip}}(k)$ is unsatisfiable for all $k$.
\label{unsatStrip}
\end{lemma}

\begin{proof}
We interpret each critical strip as a circuit that outputs the weighted sum of the input variables in circuits $L^{xy}$ and $R^{yx}$. The assignment to $\mathbf{e}$ demands that the difference between the critical strip outputs is precisely $2^k$. But by $t^{xy}_{i,j}=t^{yx}_{j,i}$, the weighted sum of the tableau variables is the same in both critical strips. The difference in the critical strip outputs is then bounded by the larger of the sums of the input carry bits to column $k-\Delta$ in the two strips. There are fewer than $n$ input carry bits for each critical strip, each of weight $2^{k-\Delta} = 2^k / n$, therefore the difference in critical strip outputs is less than $2^k$, violating the assignment to $\mathbf{e}$.
\end{proof}

Observe that this proof only relied on the relation $t^{xy}_{ij} = t^{yx}_{ji}$ in the tableau variables. The additional requirement that the tableau variables came from an assignment to $\mathbf{x,y}$ is unnecessary to refute $\phi_{\mathrm{Strip}}(k)$.

\begin{lemma}
\label{tabProof}
There is an $O(k^7 \log k)$-sized ordered resolution proof that $\phi_{\mathrm{Strip}}(k)$ is unsatisfiable.
\end{lemma}

\begin{proof}
For simplicity we assume that $k \leq n$; the case where $k > n$ is similar. We will also preprocess $\phi_{\mathrm{Strip}}(k)$ by resolving on the variables in $\mathbf{x,y}$ to obtain the tableau variable relations $t^{yx}_{j,i} = t^{xy}_{i,j}$, then replacing all the variables $t^{yx}_{j,i}$ by $t^{xy}_{i,j}$ in the clauses $\phi_{\mathrm{Strip}}(k)$. Viewing the proof as a branching program, this amounts to querying $\mathbf{x,y}$ at the end. We will not resolve on $\mathbf{x,y}$ in the remainder of this proof.

We give this resolution proof in the form of a labeled read-once branching program $B$. We define the \emph{input variables} $\sigma_{input}$ as the set of tableau variables of circuit $L^{xy}$, together with the carry variables from column $k-\Delta-1$ of both $L^{xy}$ and $R^{yx}$. We say $\sigma_{input}$ contains the \emph{input variables} to this critical strip, since their values determine an output assignment.

The idea behind the branching program $B$ is to verify circuit $L^{xy}$ by branching on its input variables row-by-row, going from top-to-bottom, remembering an assignment to a row of sum-variables. Since $t^{xy}_{i,j} = t^{yx}_{j,i}$, the tableau variables of circuit $R^{yx}$ simultaneously are revealed from bottom to top. In circuit $R^{yx}$ we maintain both a guess for its output values, and a row of sum-variables. From the proof of Lemma~\ref{unsatStrip}, if we have found that the outputs of $L^{xy}$ and $R^{yx}$ were computed correctly then they must violate one of the constraints $e_{k}=0, \ldots,e_{k-\Delta+1}=0, e_{k-\Delta}=1$.

\begin{definition}
Define $\cut(0)$ as the set of variables containing
\[  d^{yx}_{0,i}, o^{yx}_{i-1} \quad \textrm{for} \quad i-1 \in [k-\Delta,k]. \]
For $j \in [1,k-\log k]$, we define $\cut(j)$ to be the set containing the variables:
\begin{align*}
 d^{xy}_{i,j-1}, d^{yx}_{j,i-1} \quad &\textrm{for} \quad i+j-1 \in [k-\Delta,k], \\
 c^{yx}_{j-1,i} \quad &\textrm{for} \quad i+j-1 \in [k-\Delta,k-1], \\
 o^{yx}_{i} \quad &\textrm{for} \quad i \in [k-\Delta,k].
\end{align*}
Lastly, for $j \in [k-\Delta,k]$, we define $\cut(j)$ to be the set containing the variables, when the indices are in-range:
\begin{align*}
 o^{xy}_{i} \quad &\textrm{for} \quad i \in [k-\Delta,j-1], \\
  d^{xy}_{i+1,j-1}, d^{yx}_{j,i},c^{yx}_{j-1,i} \quad &\textrm{for} \quad i+j \in [k-\Delta,k], \\
 c^{yx}_{j-1,i} \quad &\textrm{for} \quad i+j-1 \in [k-\Delta,k-1], \\
 o^{yx}_{i} \quad &\textrm{for} \quad i \in [k-\Delta,k]. 
\end{align*}
\end{definition}
We will label each node of $B$ by the pair $(\cut(j),\sigma)$ where $\cut(j)$ keeps track of the previously seen cut.

\begin{figure}[t]
  \centering 
\includegraphics[width=0.9\textwidth]{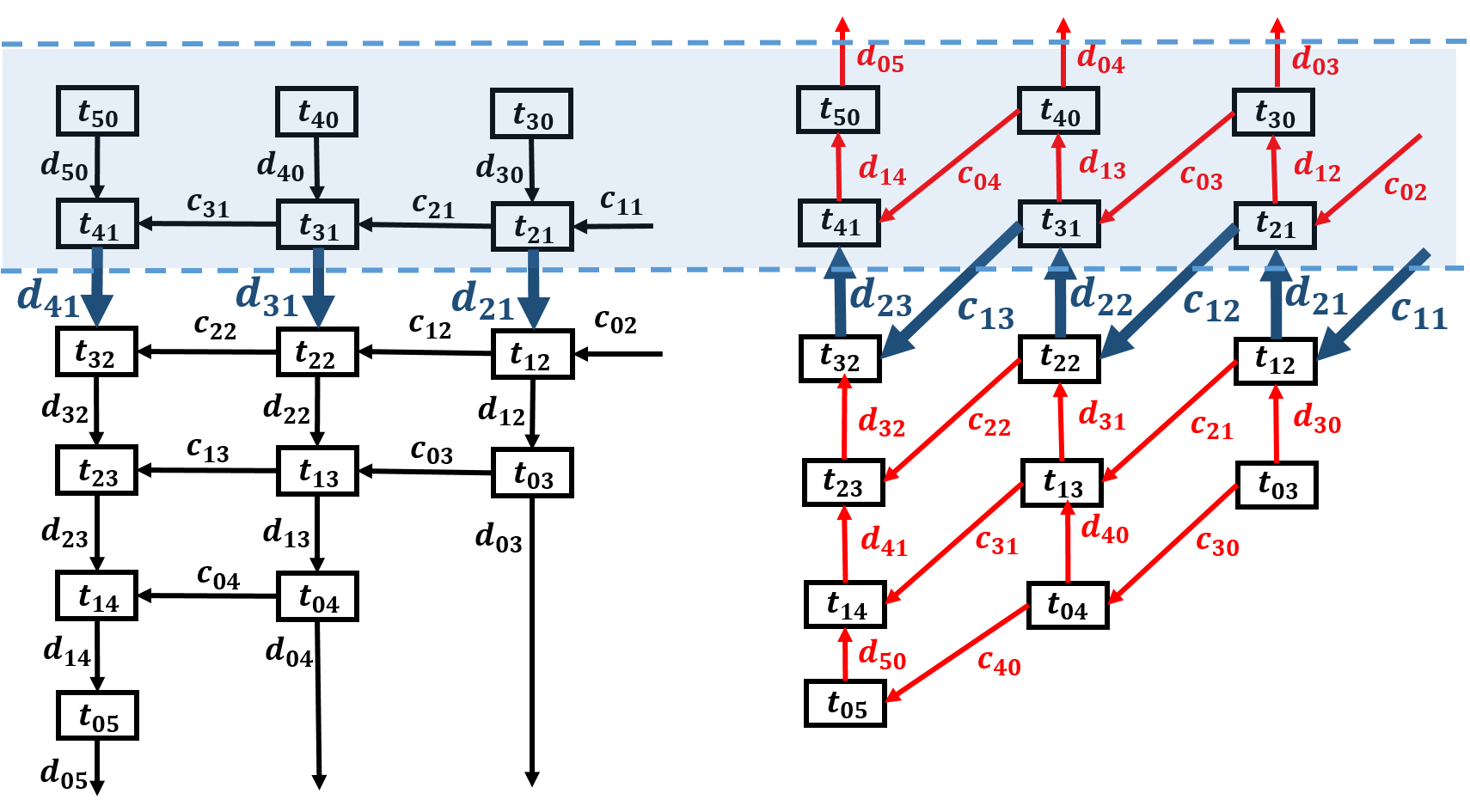}
\caption{The critical strip $\strip(5)$ for checking commutativity. The enlarged variables belong to $\cut(2)$ of $\strip(5)$. This cut divides the critical strip into a shaded satisfiable region and an unshaded unsatisfiable region.}
\label{fig:commCut}
\end{figure}

\paragraph{Initialization:} Throughout, we work in terms of the tableau variables in circuit $L^{xy}$, implicitly substituting $t^{xy}_{ij}$ for $t^{yx}_{ji}$. We begin at the root node of the read-once branching program $B$, labeled with an empty cut and an empty partial assignment $(\emptyset, \emptyset)$. For $i \in [k-\Delta,k]$ we branch on the variable $o^{yx}_i$, then propagate to $d^{yx}_{0,i}$ using a clause from the constraint $o^{yx}_i = d^{yx}_{0,i}$. The surviving branches are those labeled by an assignment satisfying the constraints $o^{yx}_i =  d^{yx}_{0,i}$. At this point we have reached nodes labeled $\cut(0)$.

For each of the surviving branches, we branch on the tableau variables in the first row of $xy$:
\[t^{xy}_{i,0} \quad \mathrm{for} \quad i \in [k-\Delta,k].\]
Then we propagate to the variables, in sequence,
\[d^{yx}_{1,i}, c^{yx}_{0,i} \quad \mathrm{for} \quad i+1 \in [k-\Delta,k] \]
from $\cut(1)$ (notice that this does not include the input carry-bit $c^{yx}_{0,k-\Delta-1}$). We then merge on $\cut(1)$.

\paragraph{Inductive Step:} We now describe the transition from $\cut(j)$ to $\cut(j+1)$ for $1\le j \le k$. Suppose that the branching program $B$ has reached an assignment to $\cut(j)$. From these nodes we branch on the next, $j$-th row's tableau variables
\[  t^{xy}_{i,j} \quad \mathrm{for} \quad i+j \in [k-\Delta,k] \]
and, when they exist, the pair of incoming input carry variables $c^{L}_{i,j}, c^{R}_{j-1,i}$ from column $k-\log k-1$.
We then propagate to the $\cut(j+1)$ and $c^L$ variables in the sequence:
\[c^{xy}_{i,j},d^{xy}_{i+1,j}  \quad \text{for} \quad i+j+1 \in [k-\log k,k]\]
in circuit $L^{xy}$. If $j \in [k-\Delta,k]$ then we also propagate to $o_{j-1}$.
\[c^{yx}_{j,i}, d^{yx}_{j+1,i}  \quad \text{for} \quad i+j+1 \in [k-\log k,k] \]
in circuit $R^{yx}$.
After branching on the last variable in $\cut(j+1)$ we start labeling nodes by $\cut(j+1)$ and merge branches on their assignment to $\cut(j+1)$. This completes the step from $\cut(j)$ to $\cut(j+1)$. 

We repeat this step until we have reached $\cut(k+1)$. At this point we have an assignment to the critical strip output bits $\mathbf{o}^{xy},\mathbf{o}^{yx}$. Furthermore, both output assignments were the result of, and therefore consistent with, propagating from a single assignment on the input variables $\sigma_{inputs}$. By the proof of Lemma~\ref{unsatStrip}, this implies that our assignment to $\mathbf{o}^{xy},\mathbf{o}^{yx}$ conflicts with an inequality constraint.

\paragraph{Size Bound:} We show that there are $O(k^6 \log k)$ nodes in $B$. Each $\cut(j)$ section of $B$ begins with an assignment to at most $4 \log k$ variables, so there are at most $k^4$ nodes labeled by an assignment to precisely $\cut(j)$.  We branch on up to $\log(k)+2$ input variables, so each cut has a full binary tree of $8k$ nodes branching on different configurations of input variables. For each leaf of this tree, $B$ has a path of $O(\log k)$ nodes for propagating before the nodes get merged. Therefore each cut labels at most $O(k^5 \log k)$ nodes. There are $k+1$ different cuts, thus $B$ has at most $O((k+1) k^5 \log k) = O(k^6 \log k)$ nodes.
\end{proof}

Since the tableau variables were actually partial products of $\mathbf{x}$ and $\mathbf{y}$, we can make this proof smaller by branching on the bits of $\mathbf{x},\mathbf{y}$ to determine the tableau variables in a row, maintaining a sliding window of $\Delta$ bits of $\mathbf{x}$, yielding:

\begin{corollary}
\label{xyproof}
$\phi_{\mathrm{Strip}}(k)$ has an $O(k^5 \log k)$-size regular resolution refutation.
\end{corollary}

We note that the alternative strategy of directly branching on the cuts to perform binary search on the critical strip yields the same size bound as Corollary~\ref{xyproof}

\begin{theorem}
Let $N = |\addstepcomm| = O(n^2)$. There is an $O(N^3 \log N)$ size regular resolution proof that $\addstepcomm$ is unsatisfiable. There is an $O(N^{7/2} \log N)$ size \emph{ordered} resolution proof that $\addstepcomm$ is unsatisfiable.
\end{theorem}

\begin{proof}
We can now describe the overall branching program $B$ for $\addstepcomm(n)$. The branching program branches on the inequality-constraint assignments $\sigma_e(k) = \{e_k=1,e_{k-1}=0,\ldots e_0=0\}$ for $k \in [0,2n-1]$. The $k$-th branch contains the clauses $\strip(k)$ so we can use the read-once branching program from either Corollary~\ref{xyproof} or Lemma~\ref{tabProof} (with each node augmented with the assignment $\sigma_e(k)$) to show that the branch is unsatisfiable. Corollary~\ref{xyproof} yields the regular resolution proof and Lemma~\ref{tabProof} yields the ordered resolution proof. 
\end{proof}

\subsection{Proofs of Array Multiplier Distributivity}

\begin{definition}
We define a SAT instance $\phi^{\mathrm {Array}}_{\mathrm{Dist}}(n)$ to verify the distributivity property
$x (y + z) = xy + xz$
for an array multiplier in the natural way. For the left hand expression we construct a ripple-carry adder $L^{y+z}$, outputting $\mathbf{o}^{(y+z)}$, and array multiplier $L^{x(y+z)}$ outputting $\mathbf{o}^{x(y+z)}$. For the right hand expression, we similarly define circuits $R^{xz}$, $R^{xy}$ and $R^{xy+xz}$.

We define $L =  L^{y+z} \cup L^{x(y+z)}$ and $R = R^{xz} \cup R^{xy} \cup R^{xy+xz}$. We let $E$ contain the usual inequality constraints. The full distributivity instance is then $\phi^{\mathrm {Array}}_{\mathrm{Dist}}(n) = L \cup R \cup E$.
\end{definition}

We again divide the instance into critical strips, following the strategy previously used to refute $\phi^{\mathrm {Array}}_{\mathrm {Comm}}$. 

\begin{definition}
Define the constant $\Delta = \log (2n)$. Let $\phi_{\mathrm{Strip}}(k)$ contain the following constraints from $\phi^{\mathrm {Array}}_{\mathrm{Dist}}(n)$: first, the full ripple-carry adder circuit $L^{y+z}$. Second, include the constraints containing one of the tableau variables $t^{{x(y+z)}}_{i,j},t^{{xy}}_{i,j}, t^{{xz}}_{i,j}$ for $i+j \in [k-\Delta,k]$. Third, include the ripple-carry adder constraints on the carry-bits and sum-bits $c^{{xy+xz}}_{i},o^{{xy+xz}}_i$ for $i \in [k-\Delta,k]$. Lastly, add constraints to $\phi_{\mathrm{Strip}}(k)$ that assign: $e_k=1, e_{k-1}=0, \ldots, e_0=0$.
\end{definition}

\begin{lemma}
\label{unsat_dist}
$\phi_{\mathrm{Strip}}(k)$ is unsatisfiable for all $k$
\end{lemma}

\begin{proof}
Like the proof of Lemma~\ref{unsatStrip}, the critical strip 
for $L^{x(y+z)}$ holds tableau bits with the same weighted sum (modulo $2^{k+1}$) 
as those in $R^{xz}$ and $R^{xy}$ combined. The critical strip for $L^{x(y+z)}$ 
has at most $n$ input carry-bits of weight $2^{k-\Delta}$. 
The critical strips of the $n$-bit multipliers $R^{xz}$ and $R^{xy}$ each have at most $n-1$ 
input carry variables of weight $2^{k-\Delta}$. The critical strip of the 
adder $R^{xy+xz}$ has one input carry variable, so the critical strip for $R$ 
has $2n-1$ input carry-bits. Since we set the width of the strip at $\Delta = \log (2n)$, it is unsatisfiable.
\end{proof}

\begin{lemma}
\label{thm:distStrip}
For each $k$ there is an $O(n^5 \log n)$ size regular resolution proof that $\phi_{\mathrm{Strip}}(k)$ is unsatisfiable.
\end{lemma}

\begin{proof}
We construct a labeled branching program $B$ that solves the conflict clause search problem for $\phi_{\mathrm{Strip}}(k).$ We branch row-by-row in the critical strips, maintaining an assignment to cuts of variables in each multiplier. For each strip we will select a (different) variable ordering for $\mathbf{x,y,z}$ that reveals the tableau variables row-by-row. Assume that $k < n$ for simplicity; the case where $k \geq n$ is similar.

For an array multiplier computing an expression $C \in \{{x(y+z)}, {xz}, {xy} \}$ and $j \in [1,k-\Delta]$ we define $\cut^C(j)$ to be the set of variables
\[d^{C}_{i,j-1} \quad \text{for} \quad i+j-1 \in [k-\Delta,k], \]
and for $j \in [k-\Delta+1,k]$ we define $\cut^C(j)$ as the set of variables
\begin{align*}
d^{C}_{i,j-1} \quad &\text{for} \quad i+j-1 \in [k-\Delta,k], \\
o^{C}_{i} \quad &\text{for} \quad i \in [k-\Delta,j-2] 
\end{align*}
We define $\cut^{y+z}(j)$ as the singleton set $\{c^{y+z}_{j-1}\}$ and define $\cut^{x}(j)$ as the set
\[x_i \quad \text{:} \quad i \in [k-j-\Delta,k-j].\].
We also refer to a global cut, across the whole circuit: $\cut(j) = \cup_C  \cut^{C}(j)$.
\begin{figure}[t]
  \centering 
\includegraphics[width=0.9\textwidth]{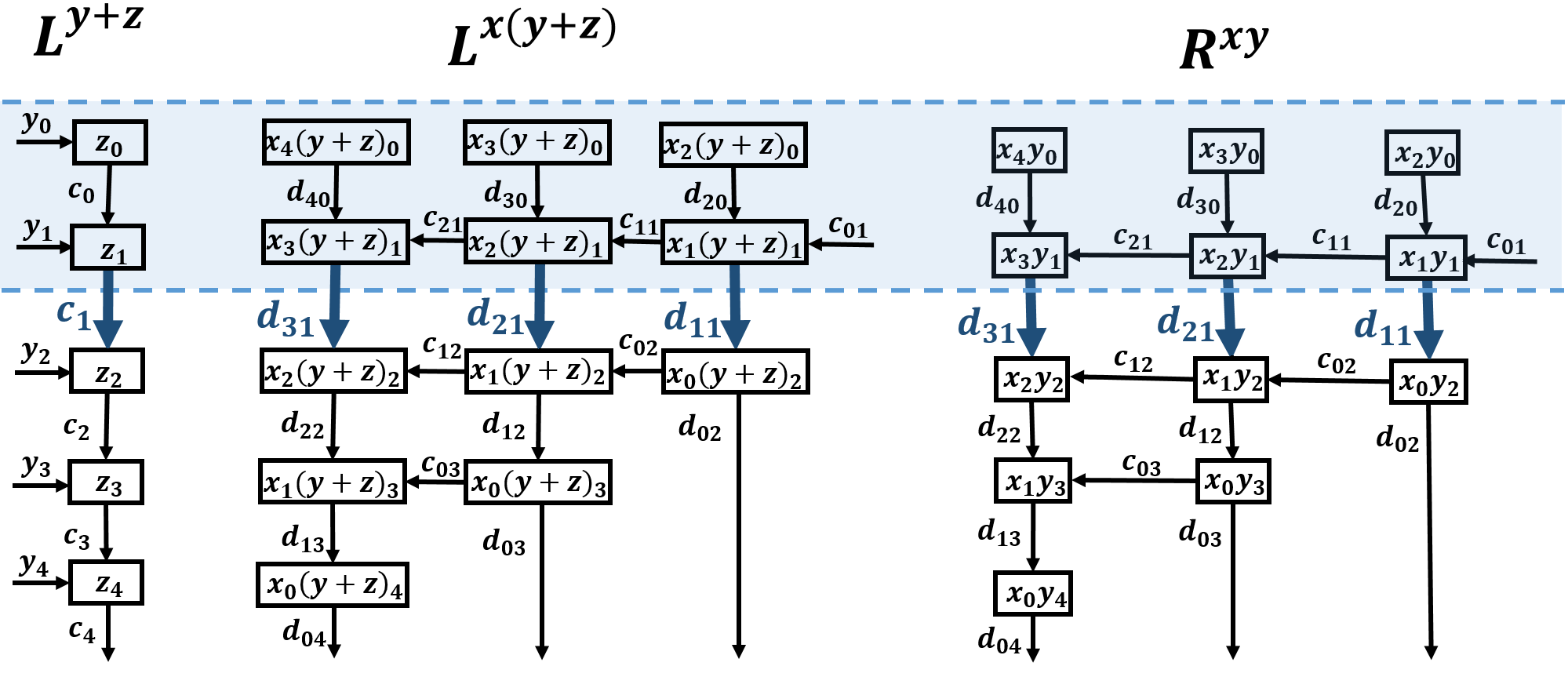}
\caption{The critical strip $\phi_{\mathrm{Strip}}(4)$ for distributivity. $\cut(2)$ consists of the enlarged variables.}
\label{fig:cutbranchingDist}
\end{figure}

\paragraph{Initialization: Getting to $\cut(1)$}

\begin{sloppypar}
At the root node $(\emptyset,\emptyset)$ of $B$, we branch on the circuit input variables $y_0,z_0$ and
\[ x_i \quad \text{for} \quad i \in [k-\Delta,k] .\]
We propagate these assignments to variables $c^{{y+z}}_0$ and $o^{{y+z}}_0$, giving us an assignment to $\cut^{{y+z}}(0)$. The assignment to $o^{{y+z}}_0$, in turn, propagates to an assignment to the first row of tableau and sum variables from the critical strip for $L^{x(y+z)}$:
\[ t^{{x(y+z)}}_{i,0}, d^{{x(y+z)}}_{i,0} \quad \text{for} \quad i \in [k-\Delta,k]. \]
At this point we have an assignment to $\cut^{{x(y+z)}}(0)$.
\end{sloppypar}

We then propagate the input variable assignments through the multipliers $R^{xy}$ and $R^{xz}$:
\[ t^{{xy}}_{i,0}, d^{{xy}}_{i,0} \quad \text{:} \quad i \in [k-\Delta,k], \]
\[ t^{{xz}}_{i,0}, d^{{xz}}_{i,0} \quad \text{:} \quad i \in [k-\Delta,k], \]
obtaining assignments to $\cut^{{xy}}(0)$ and $\cut^{{xz}}(0)$, thus completing an assignment to $\cut(0)$. At this point we merge nodes on assignment to $\cut(0)$.

\paragraph{Inductive Step: $\cut(j)$ to $\cut(j+1)$}

Suppose we have merged branches and are at a node labeled with an assignment to $\cut(j)$. If this assignment contains a variable $d^C_{0,i}$ we propagate to $o^C_{i}$. We branch on input variables 
$x_{k-\Delta-j}, y_{j}, z_{j}$. We then propagate these assignments to $c^{{y+z}}_{j+1}, o^{{y+z}}_{j+1}$, followed by the next row of tableau, carry, and sum variables in each multiplier:
\[ c^{C}_{i-j-2,j+1}, t^{C}_{i-j-1,j+1}, d^{C}_{i-j-1,j+1} \quad \text{:} \quad i \in [k-\Delta,k]. \]
At this point we have reached an assignment to all of the variables in $\cut(j+1)$ so we merge nodes based on $\cut(j+1)$. We repeat this step until reaching an assignment to $\cut(k+1)$, which consists of each multiplier's output bitvector $\mathbf{o}^C$.

\paragraph{End: Beyond $\cut(k+1)$}

Suppose that we have reached $\cut(k+1)$ and merged nodes. We branch on the input carry variable $c^{{xy+xz}}_{k-\Delta-1}$, that goes into the critical strip of ripple-carry adder $R^{xy+xz}$. We can then propagate to the outputs $\mathbf{o}^{{xy+xz}}$. We now have an assignment to both $\mathbf{o}^{{x(y+z)}},\mathbf{o}^{{xy+xz}}$ that was propagated from one assignment to the input variables to the critical strip. By Lemma~\ref{unsat_dist}, this assignment conflicts with an inequality-constraint from $E$.

\paragraph{Size Bound:} There are $k+1$ different global cuts $\cut(j)$. Each $\cut(j)$ section of $B$ begins with an assignment to at most $4\Delta+1$ variables. So each section $\cut(j)$ is initialized with at most $2^{4\Delta+1} = 8n^4$ branches. Each of these branches is a path with at most $O(\log n)$ queried variables and therefore at most $O(\log n)$ nodes. So there are at most $ O(n^4 \log n)$ nodes per cut and therefore at most $O((k+1) n^4 \log n) = O(n^5 \log n)$ nodes in $B$.
\end{proof}

\begin{theorem}
\label{thm:arrayDist}
There is an $O(n^{6} \log N)$ size resolution proof that $\phi^{\mathrm {Array}}_{\mathrm{Dist}}(k)$ is unsatisfiable.
\end{theorem}

\begin{proof}
At the root of this proof there are $2n$ branches each holding an assignment to $e_k,\ldots,e_1,e_0$. \update{We refute each branch using the $O(n^5 \log n)$ size proof from Lemma~\ref{thm:distStrip}.}
\end{proof}

\subsection{Proofs of $x(x+1) = x^2 + x$ for Array Multipliers}

\begin{definition}
\label{def:xpl1}
We define a SAT instance $\phi^{\textrm{Array}}_{x(x+1)}(n)$. Circuit $L$ is composed of circuits $L^{x+1}$, consisting of a ripple-carry adder taking inputs $\mathbf{x}$ and $\mathbf{1}$ and outputting their sum $\mathbf{(x+1)}$, and $L^{x(x+1)}$, an array multiplier outputting the product $\mathbf{x(x+1)}$. Similarly, circuit $R$ is composed of circuits $R^{x^2}$ and $R^{x^2+x}$.

We let $E$ contain the usual inequality-constraints. The instance is then
\[\phi^{\textrm{Array}}_{x(x+1)}(n) = L \cup R \cup E. \]
\end{definition}

While this identity looks like a special case of distributivity, its resolution proof is more complicated. This is because for distributivity: $x(y+z) = xy + xz$, the inputs to each multiplier were separate variables. This allowed us to scan the critical strip from one end to the other in a read-once fashion. If we try a similar strategy to scan the critical strip for the multiplier $R_{x^2}$ from top to bottom, we will read each $x_i$ twice. To avoid reading the same variable twice, we instead scan the critical strip from both ends, meeting in the middle.

\begin{definition}
Define the constant $\Delta = \log (2n-1)$. Let $\phi_{\textrm{Strip}}(k)$
contain the full ripple-carry adder circuit $L^{x+1}$ from $\phi^{\textrm{Array}}_{x(x+1)}(n)$. Also include the constraints containing one of the multiplier tableau variables $t^{x(x+1)}_{i,j},t^{x^2}_{i,j}$ for $i+j \in [k-\Delta,k]$. Further include the constraints on the ripple-carry adder carry-bits and sum-bits $c^{{x^2+x}}_{i},d^{{x^2+x}}_i$ for $i \in [k-\Delta,k]$. Lastly, add constraints to $\phi_{\textrm{Strip}}(k)$ that encode the values of the bits: $e_k=1, e_{k-1}=0, \ldots, e_0=0$.

We refer to the subcircuit $\phi_{\textrm{Strip}}(k) \cap C$ as the \emph{critical strip} for $C$. Figure~\ref{fig:xpl1cut} shows an example of a critical strip.
\end{definition}

\begin{lemma}
\label{unsat_xpl1}
$\phi_{\textrm{Strip}}(k)$ is unsatisfiable for all $k$
\end{lemma}

\begin{proof}
The proof is the same as the proof for Lemma~\ref{unsat_dist}.
\end{proof}

\begin{definition}
For an array multiplier computing the expression $C \in \{x(x+1), x^2\}$ and $j \in [1,(k-\Delta)/2]$ we define $\cut^C(j)$ to be the set of variables
\[d^{C}_{i,j-1} \quad \text{:} \quad i+j-1 \in [k-\Delta,k], \quad \textrm{(upper cut)} \]
\[ c^{C}_{j-1,i}, d^{C}_{j,i} \quad \text{:} \quad i+j \in [k-\Delta,k]. \quad \textrm{(lower cut)} \]
We define $\cut^{{x+1}}(j)$ to contain $x_{j-1}$ and the set of variables
\[x_i \quad \text{:} \quad i \in [k-j-\Delta,k-j].\]
\end{definition}

\begin{figure}[t]
\centering 
\includegraphics[width=0.9\textwidth]{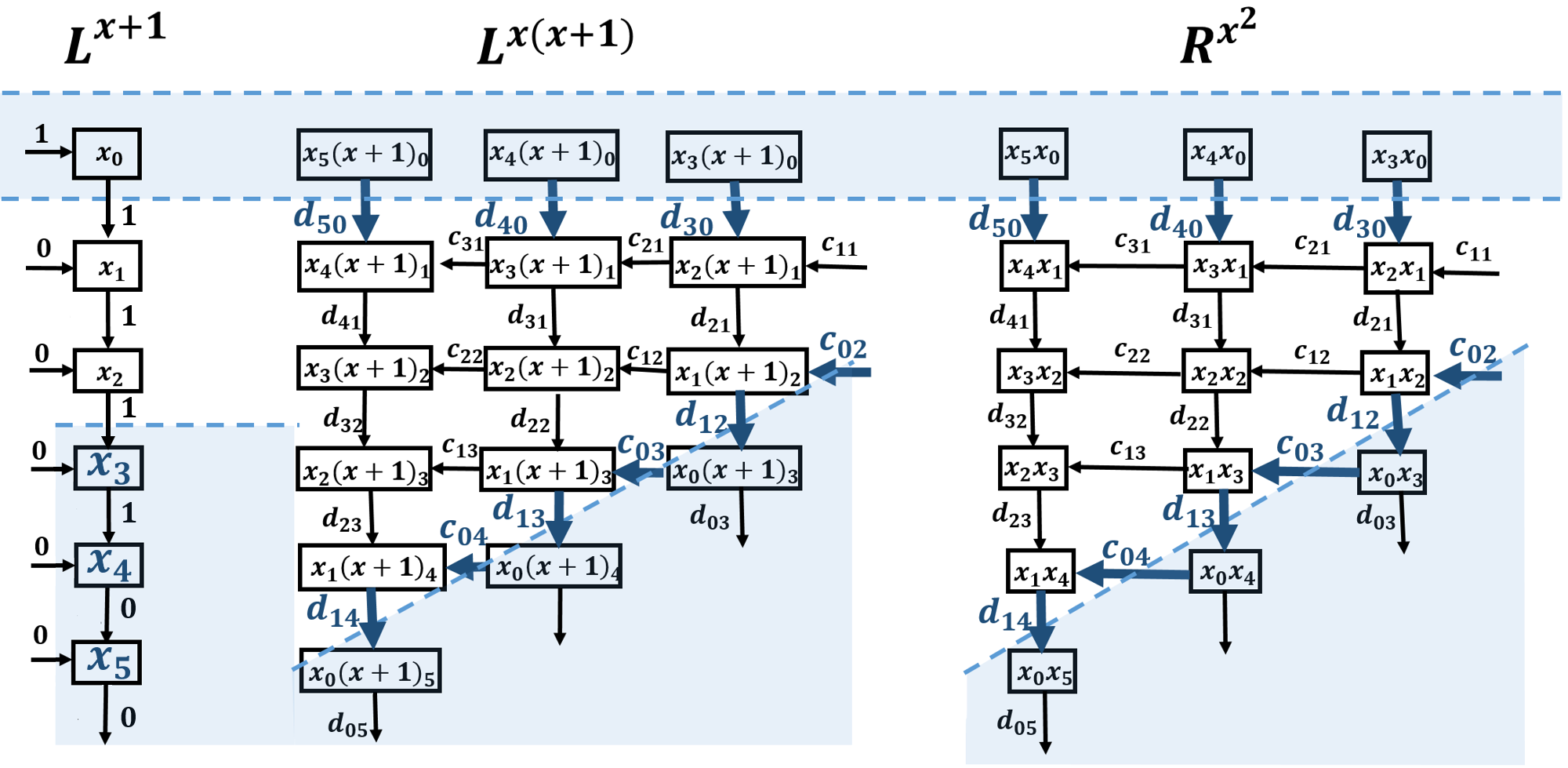}
\caption{The critical strip $\phi_{\textrm{Strip}}(5)$ for checking $x(x+1) = x^2 + x$. The shaded region is satisfiable. The enlarged variables belong to $\cut(1)$.}
\label{fig:xpl1cut}
\end{figure}

\begin{theorem}
There is a size $n^7 \log n$ regular resolution proof 
that $\phi_{\textrm{Strip}}(k)$ is unsatisfiable.
\end{theorem}
\begin{proof}

\emph{Initialization.}
We give our proof in the form of a labeled read-once branching program $B$. We begin by branching on a guess for the critical strip outputs $\mathbf{o}^{{x(x+1)}}, \mathbf{o}^{{x^2+x}}$. For the branches that don't conflict with an inequality-constraint, we branch on the values
\[o^{{x^2}}_i, x_i \quad : i \in [k-\Delta,k], \]
then merge to erase the assignment to $\mathbf{o}^{{x^2+x}}$.

We observe that the carry variables in $L^{x+1}$ must be a sequence of $1$s followed by $0$s. If, on the contrary, we observe the assignments $c_i = 0$ and $c_j = 1$ for $i < j$, then we can efficiently find a conflict by propagating $c_i = 0$ through columns $[i,j]$. So we can begin this proof by branching on the at most $n$ valid carry-bit assignments
\[ c^{x+1}_0 = 1, \ldots,c^{x+1}_i = 1, c^{x+1}_{i+1} = 0,\ldots, c^{x+1}_{k}=0 .\]

Our branch order begins on the input variables $x_0$ and $x_k,x_{k-1},\ldots,x_{k-\Delta}$. We propagate the resulting assignment to the upper and lower cuts in each circuit, then merge on the assignment to $\cut(1)$.  

\paragraph{Inductive Step}

To get from $\cut(j)$ to $\cut(j+1)$, we branch on input variables $x_{j}, x_{k-j-\Delta+1}$, then propagate to and merge on $\cut(j+1)$.

We have two cases: the upper and lower cuts of $\cut(j+1)$ either intersect or they do not. In either case we branch on input variables $x_{j-1},x_{k-\Delta-j+1}$ and the input carry variables to rows $j$ and $(k-j-\Delta+1)$.  If the cuts do not intersect, we propagate to, then merge on, all the $\cut(j+1)$ variables. Otherwise, suppose that the upper and lower cuts of $\cut(j+1)$ intersect on $d_{i,j}$. The upper and lower cuts of $\cut(j)$ either propagate to conflicting values of $d_{i,j}$, in which case we have found a conflict, or they agree on the value of $d_{i,j}$, in which case we delete column $i+j$ from our cuts.

\paragraph{Size Bound}

Each cut belongs to one of up to $n$ branches for the carry variables in $L^{x+1}$ and holds an assignment to at most $7 \log n$ variables so there are at most $n^8$ initial nodes for each cut. Each of these nodes propagates for $O(\log n)$ steps to get to the next cut, so our branching program has size $O(n^9 \log n)$.
\end{proof}

We can now obtain a refutation for $\phi^{\textrm{Array}}_{x(x+1)}(n)$ by branching on sequences of variables in $\mathbf{e}$ and using the refutation for $\phi_{\textrm{Strip}}(k)$ on each branch.

\begin{theorem}
There is a size $n^{10} \log n$ regular resolution proof that the SAT instance $\phi^{\textrm{Array}}_{x(x+1)}(n)$ is unsatisfiable.
\end{theorem}

\subsection{Degree Two Identity Proofs for Array Multipliers}

Let $\phi^{\mathrm{Array}}_{L=R}(n)$ denote a SAT instance checking that the array multiplier obeys the ring identity $L=R$. With the insight from the earlier proofs in this section, we can prove the general theorem:

\begin{theorem}
\label{thm:deg2array}
For any degree two ring identity $L=R$, there are polynomial size regular refutations for $\phi^{\mathrm{Array}}_{L=R}(n)$.
\end{theorem}

\begin{proof}
(Sketch) We divide $\phi^{\mathrm{Array}}_{L=R}(n)$ into unsatisfiable critical strips of width $\Delta = \log mn$, where $m$ is the number of terms in the identity $L=R$. The ripple-carry adders that input to a multiplier remain intact, and for the rest we remove the columns outside the critical strip.

We begin by branching on guesses for the $\Delta$ output bits from each multiplier and each truncated ripple-carry adder. In each multiplier we use a "meet-in-the-middle" strategy, similar to the proof for $x(x+1) = x^2+x$. We read all the input bitvectors in parallel, each in the same order. This branch order for each input bitvector $\mathbf{x}$ is $x_0,x_n,x_1,x_{n-1},\ldots$.
We branch on the input carry-bits as needed to propagate the cuts. We can propagate the resulting input variable assignments to diagonal cuts in each multiplier that scan from the top and bottom edges towards the middle, and likewise for the intact ripple-carry adders. In each input bitvector we remember the assignment to just the most recently queried $2\Delta$ variables. Because of the symmetry of this variable order, it is compatible with swapping the order of inputs to any multiplier, as well as multipliers squaring an input.
\end{proof}

\section{Diagonal Multipliers and Booth Multipliers}
\label{sec:diagonal}
A diagonal multiplier uses a similar idea to the array multiplier. 
The difference is that the diagonal multiplier routes its carry bits to the next row instead of the same row as depicted in Figure~\ref{fig:diagonal}. 

A Booth multiplier uses a similar idea to the array multiplier, but uses two's complement notation and a telescoping sum identity to skip consecutive digits in one multiplicand. To add the terms of this sum, the Booth multiplier uses a grid of full adders similarly to the array multiplier,  but with some small modifications to accommodate signed integers.

Like with the array multiplier, we can divide the diagonal and Booth multipliers into $O(\log n)$-width unsatisfiable critical strips. Using the same input variable orderings from~Section~\ref{sec:array} we can verify each of these critical strips with a polynomial-size regular resolution proof.

\begin{definition}
Let $\phi^{\mathrm{Diag}}_{L=R}(n)$ denote the SAT instance checking that an $n$-bit diagonal multiplier obeys the ring identity $L=R$. Likewise let $\phi^{\mathrm{Booth}}_{L=R}(n)$ denote the SAT instance checking that an $n$-bit Booth multiplier obeys the ring identity $L=R$
\end{definition}

\begin{theorem}
For any degree two ring identity $L=R$, there are polynomial size regular resolution proofs for $\phi^{\mathrm{Diag}}_{L=R}(n)$ and $\phi^{\mathrm{Booth}}_{L=R}(n)$
\end{theorem}

\begin{proof}
(Sketch) We divide $\phi^{\mathrm{Diag}}_{L=R}(n)$ or $\phi^{\mathrm{Booth}}_{L=R}(n)$ into unsatisfiable critical strips of width $\Delta = \log mn$, where $m$ is the number of terms in the identity $L=R$.
This is the same width as in the array multiplier since the number of 
input carry-bits in each multiplier's critical strip is at most $n$. The ripple-carry adders that input to a multiplier remain intact, and for the rest we remove the columns outside the critical strip. We note that although the Booth multiplier uses two's complement signed integers, this does not materially affect our critical strip proofs.

We begin by branching on guesses for the $\Delta$ output bits from each multiplier and each truncated ripple-carry adder. We use the same branch order as in the array multiplier proof: each input bitvector $\mathbf{x}$ is read in parallel, in the order $x_0,x_n,x_1,x_{n-1},\ldots$. We branch on the input carry-bits as needed to propagate the cuts. We can propagate the input variable assignments to diagonal cuts in each multiplier that scan from the top and bottom edges towards the middle, and likewise for the intact ripple-carry adders. In each input bitvector we remember the assignment to just the most recently queried $2\Delta$ variables.
\end{proof}

\section{Wallace Tree Multipliers}
\label{sec:wallace}
\subsection{Wallace Tree Multiplier Construction}
\label{sec:wal_construct}

A Wallace tree multiplier takes a different approach to summing the tableau. Using carry-save adders (parallel 1-bit adders), it iteratively finds a new tableau with the same weighted sum as the previous tableau, but with $1/3$ fewer rows. Upon reducing the original tableau to just two rows, it uses a carry-lookahead adder to obtain the final result. In contrast to the array multiplier, a Wallace tree multiplier is complicated to lay out physically, but has only logarithmic depth. 

\paragraph{Carry-Lookahead Adder:}
\begin{figure}[t]
  \centering 
\includegraphics[width=0.9\textwidth]{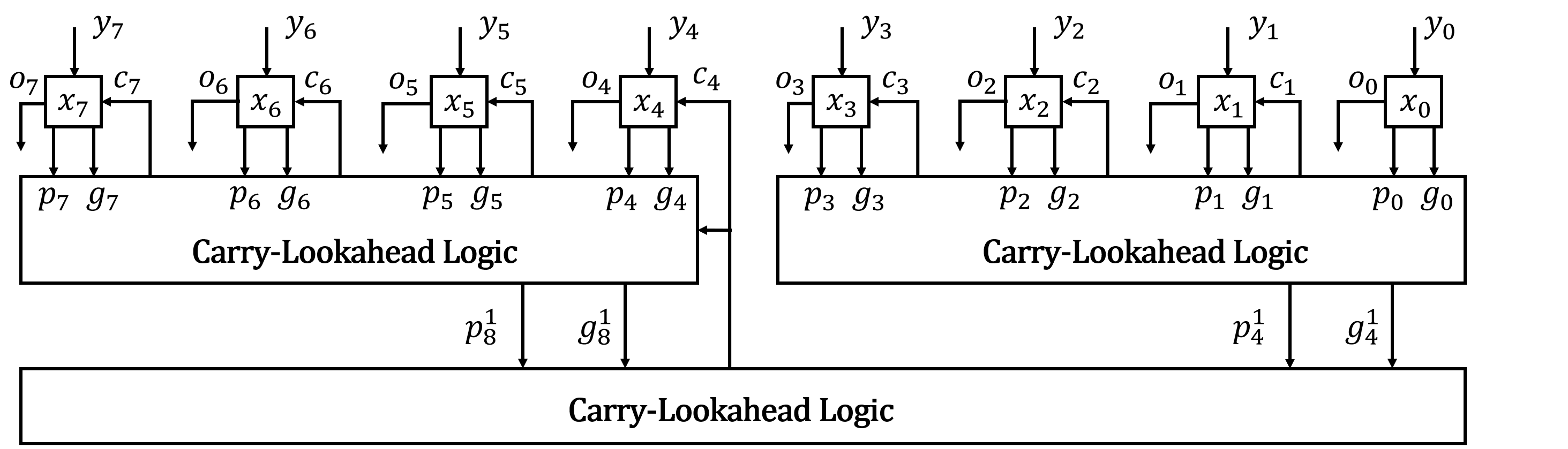}
\caption{8-bit, two-layer CLA adding $\mathbf{x,y}$.}
\label{fig:CLA}
\end{figure}
A carry-lookahead adder (CLA) uses a tree structure to add two bitvectors $\mathbf{x,y}$
with only logarithmic depth. The 4-bit CLA computes, for each pair $x_i,y_i$, the values
\[ g_{i} = x_i y_i \quad \quad p_{i} = x_i \oplus y_i. \]
Then, writing $c_i$ for the carry bit in the $i$-th column, we have
\[ c_{i+1} = g_{i} \oplus (p_{i} c_i) .\]
We can use this to derive the following equations, which we can use to compute each carry digit in parallel from the values $g_i,p_i$ and $c_0$:
\begin{eqnarray*}
\label{carry_eqns}
c_1 &=& g_{0} \oplus p_{0} c_0 \\
c_2 &=& g_{1} \oplus g_{0} p_{1} \oplus c_0 p_{0,0} p_{1} \\
c_3 &=& g_{2} \oplus g_{1} p_{2} \oplus g_{0} p_{1} p_{2} \oplus c_0 p_{0} p_{1} p_{2} \\
c_4 &=& g_{3} \oplus g_2 p_3 \oplus g_1 p_2 p_3 \oplus g_0 p_1 p_2 p_3  \oplus c_0 p_0 p_1 p_2 p_3.
\end{eqnarray*}
These values are used to compute the outputs: $o_i = c_i \oplus x_i \oplus y_i$.
It additionally computes the \emph{group propagate} and \emph{group generate}:
\begin{eqnarray*}
p_{1,4} &=& p_{3} p_{2} p_{1} p_{0} \\
g_{1,4} &=& g_{3} \oplus g_{2} p_{3} \oplus g_{1} p_{3} p_{2} \oplus g_{0} p_{3} p_{2} p_{1},
\end{eqnarray*}
where the first index indicates the layer.

We construct a $16$-bit CLA with $2$ layers, whose first half of is shown in Figure~\ref{fig:CLA}. At the zero-th layer we arrange four 4-bit CLAs, the $k$-th CLA taking inputs $x_i,y_i, i \in [4k,4k+3]$ and outputting to $p_{0,i},g_{0,i}, i \in [4k,4k+3]$, where the superscript indicates the layer. We denote the $k$-th CLA group propagate and generate by $p_{1,4k} g_{1,4k}$. Then the carries $c_4,c_8,c_{12},\ldots$ can be computed by the equations
\begin{eqnarray*}
c_4 &=& g_{1,0} \oplus p_{1,0} c_0 \\
c_8 &=& g_{1,4} \oplus g_{1,0} p_{1,4} \oplus c_0 p_{1,0} p_{1,4} \\
c_{12} &=& g_{1,8} \oplus g_{1,4} p_{1,8} \oplus g_{1,0} p_{1,4} p_{1,8} \oplus c_0 p_{1,0} p_{1.4} p_{1,8} \\
c_{16} &=& g_{1,12} \oplus g_{1,8} p_{1,12} \oplus g_{1,4} p_{1,8} p_{1,{12}} \oplus p_{1,0} p_{1,4} p_{1,8} p_{1,12}  \oplus c_0 p_{1,0} p_{1,4} p_{1,8} p_{1,12}. \\
\end{eqnarray*}
Notice that these equations are isomorphic to the previous equations for computing carries within each 4-bit CLA. We can reuse the same circuitry from the 4-bit CLA to compute these carries, as well as the group propagate and generate for the next layer.
We can repeat this process to construct larger CLAs, with each iteration able to handle four times the bitwidth.

\paragraph{Wallace Tree Multiplier:}

\begin{figure}[t]
  \centering 
\includegraphics[scale = 0.22]{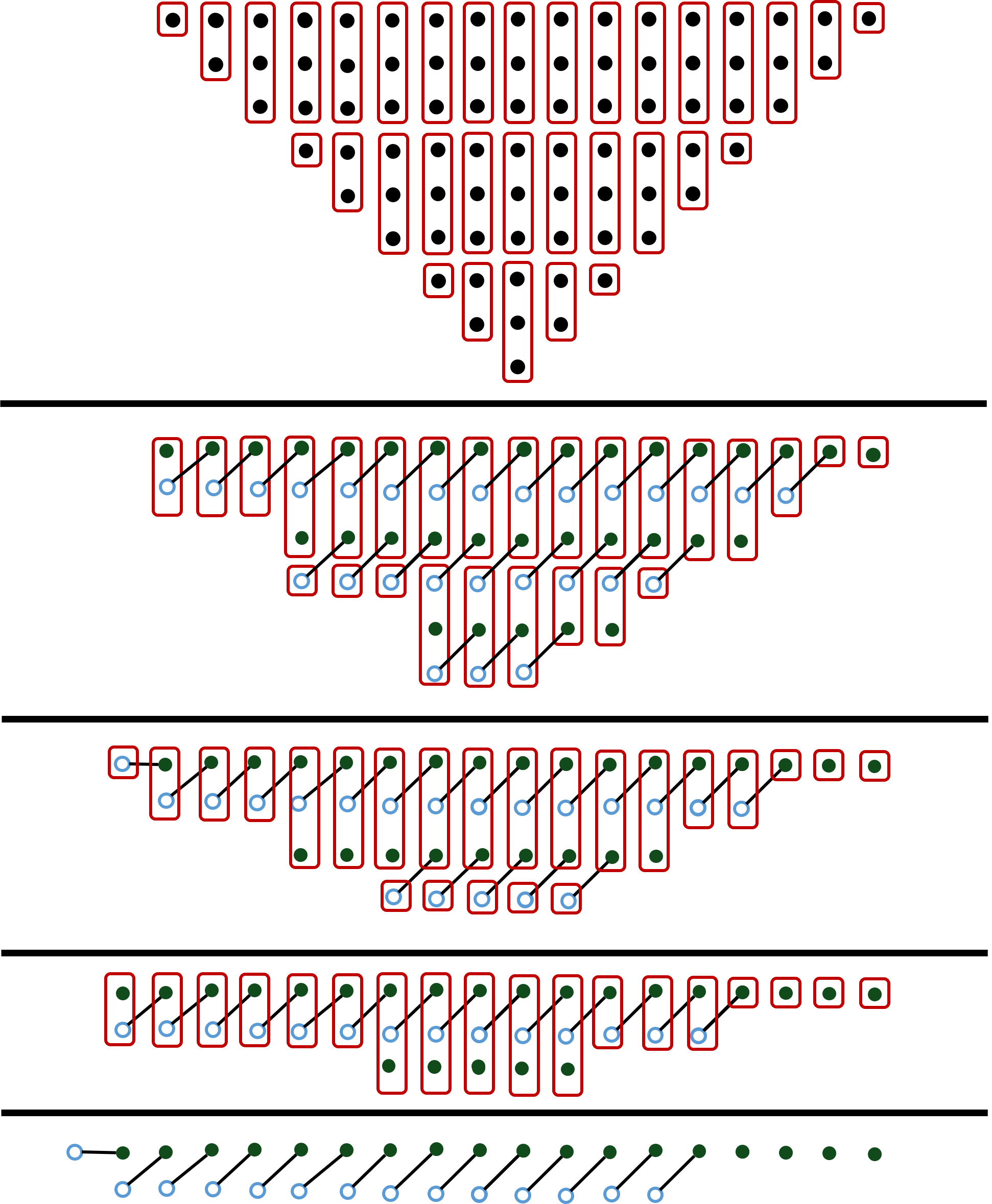}
\caption{Dot diagram for a $9\times 9$ Wallace tree multiplier. Hollow dots represent carry-bits and solid dots represent sum-bits. Dots connected by an edge are output by the same adder.}
\label{fig:wallaceDotDiagram}
\end{figure}

We construct a Wallace tree multiplier taking input $(\mathbf{x},\mathbf{y})$. We compute a tableau of partial products like in the array multiplier. We then go through $h \approx \log n$ steps to reduce the $n$-row starting tableau to an equivalent $2$-row tableau.

We define \emph{tableau variables} $t_{\ell,i,j}$ where $\ell$ is the layer of the tableau, $i$ is the index of the column containing the adder and $j$ is the row. We will denote the set of tableau variables in a column by
\[ \col(i) = \{t_{\ell,i,j} \quad \textrm{for all} \quad \ell,j \}, \]
and call the subset of a column within a layer $l$ 
a \emph{subcolumn}, denoted by
\[ \col(\ell,i) = \{t_{\ell,i,j} \quad \textrm{for all} \quad j \}. \]
In the zero-th layer, the tableau variables represent the partial products:
\[ t_{0,i,j} = x_{i-j} \wedge y_{j}  \quad \textrm{for} \quad i < n,\]
\[ t_{0,i,j} = x_{n-1-j} \wedge y_{i-n+j+1}  \quad \textrm{for} \quad i \geq n.\]
We now specify how to construct layer $\ell+1$ from layer $\ell$. We partition the rows of layer $\ell$ into sets of three, from top to bottom. Adder $A_{\ell,i,j}$ will take input from the $i$-th column of the $j$-th set of three rows. For each row of adders $j = 0, 1, \ldots$, for each $i \in [0,2n]$, we append adder $A_{\ell,i,j}$'s sum-bit to subcolumn $\col(\ell+1,i)$. Then for each $i$, we append adder $A_{\ell,i,j}$'s carry-bit to subcolumn $\col(\ell+1,i+1)$.

Each layer reduces the number of rows in the tableau from $N$ to $\lceil 2N/3 \rceil$. The tableau for the last layer $h < \log_{3/2}(n) < 2 \log n$, will only have two rows. We use a $2n$-bit\footnote{This is not a $(2n-1)$-bit adder because the top summand may have $2n$ bits.} carry-lookahead adder (CLA) to sum the two rows in logarithmic depth, outputting the final sum in the output bits $o_i$.

Like the proofs for array multipliers, our proofs for Wallace tree multipliers
divide the instance into critical strips. 
In fact, our proofs branch on the input tableau in the same row-by-row order in both array and Wallace tree multipliers. 
However the size of the resulting cuts is $O(\log ^2 n)$ for Wallace tree multipliers rather than the
$O(\log n)$ size cuts for array multipliers. 
This cut size results in quasipolynomial size regular resolution
proofs. 

When analyzing the cuts in a Wallace tree multiplier, we will find the following property useful:

\begin{definition}
For layer $\ell$ of a Wallace tree multiplier, if for each $j \leq k$, the outputs of $j$-th row of adders, $\{A_{\ell,i,j}\}_i$, map to and cover the rows $2j,2j+1$ of the next layer $\ell+1$'s tableau, we say that layer $\ell$ is \emph{row-friendly} up to its $k$-th row of adders. If layer $\ell$ is row-friendly up to its last row of adders, we say that layer $\ell$ is \emph{row-friendly}.
\end{definition}

\begin{lemma}
\label{thm:nice_adders}
% Fixes a bug in amsbook
\label{THM:NICE_ADDERS}
In a Wallace tree multiplier, each layer $\ell \in [0,h-2]$ is row-friendly.
\end{lemma}

In terms of the dot diagram in Figure~\ref{fig:wallaceDotDiagram}, this Lemma simply states that no two bits are connected with a line of slope greater than one.

\subsection{Proofs of Wallace Tree Multiplier Commutativity}

\begin{definition}
We define a SAT instance $\phi^{\mathrm {Wall}}_{\mathrm{Comm}}(n)$. The inputs to the multipliers are $n$-bit integers $\mathbf{x,y}$. Using the construction from Section~\ref{sec:wallace}, we define Wallace tree multipliers $L$, computing $\mathbf{xy}$, and $R$, computing $\mathbf{yx}$ (reversing the order of multiplier inputs).

After specifying the circuits $L$ and $R$, we add a circuit $E$, of of \emph{inequality-constraints} encoding that the two circuits disagree on some output bit.
\end{definition}

\begin{definition}
Define $\delta = \log(n+2)$. Let $\phi_{\textrm{Strip}}(k)$ contain
the constraints from $\phi^{\mathrm {Wall}}_{\mathrm{Comm}}(n)$ that contain a tableau variable $t^{xy}_{\ell,i,j}$ or $t^{yx}_{\ell,i,j}$ for $i \in [k-\delta,k]$, and also the constraints for the full CLAs at the end of the Wallace tree multipliers. Also add unit clauses to $\phi_{\textrm{Strip}}(k)$ for the assignment: $e_0=0, e_1=0, \ldots,e_{k-1}=0, e_k=1$.

We call the newly unconstrained tableau bits in column $k-\delta$, that were carry-bits output by adders from the removed column $k-\delta-1$, the \emph{input carry-bits} to $\phi_{\textrm{Strip}}(k)$.
\end{definition}

\begin{lemma}
$\phi_{\textrm{Strip}}(k)$ is unsatisfiable for all $k$.
\label{wallace_unsat}
\end{lemma}

\begin{proof}
  We reason similarly to the proof of Lemma~\ref{unsatStrip}. Again, we interpret the critical strip as a circuit that computes the weighted sum, in both $L$ and $R$, of the tableau variables within the strip. The assignment to $\mathbf{e}$ asserts that the outputs of $L$ and $R$ differ by precisely $2^k$. We bound the admissible difference in outputs by counting the number of input carry-bits in either $L$ or $R$. Since each layer of a Wallace tree multiplier has $\lceil 2/3 \rceil$ fewer rows than the previous layer, the total number of tableau rows past the initial layer is at most $2n$. At most half of these rows are composed of carry-bits, so circuits $L$ and $R$ each have at most $n$ input carry-bits coming from the removed column $k-\delta-1$. Additionally, the newly unconstrained inputs to the final CLA from the removed columns can contribute a total weight of at most $2^{k-\delta}$ to the final output. Since we set $\delta = \log(n+2)$, the total difference between the final outputs is at most $2^{k-\delta}(n+2) < 2^k$.
\end{proof}

\begin{lemma}
\label{thm:walStrip}
There is a regular resolution proof of size $2^{8 \log^2 n + O(\log n)}$ that $\phi_{\textrm{Strip}}(k)$ is unsatisfiable.
\end{lemma}

\begin{proof}
The idea of this proof is to read the initial layer of the critical strip row-by-row. If we have assigned all of the inputs to a row of adders, we propagate to their output bits. In this way, an input assignment to $\mathbf{x}$ and $\mathbf{y}$ will propagate through the layers of the Wallace tree multiplier in parallel, then finally reach an assignment to the output bits of both circuits. From the proof of~\ref{wallace_unsat}, the result will contradict one of the inequality-constraints from $\phi^{\mathrm {Wall}}_{\mathrm{Comm}}(n)$. 

Each node of the branching program will only keep track of a constant number of variables in each subcolumn. This will ensure that the cuts have $O(\log^2 n)$ variables, so that the branching program has at most $2^{O(\log^2 n)}$ nodes.

We first preprocess the constraints to obtain the equalities $t^{xy}_{0,i,j} =  t^{yx}_{0,i,i-j}.$
Like in the array multiplier case, as we branch from the top tableau row downwards in circuit $L$, we will reveal the bottom row upwards in circuit $R$. We will first describe how the branching program $B$ propagates an assignment from the initial tableau to an assignment to the last layer in circuit $L$. The propagation in circuit $R$ works symmetrically, going from the bottom row of adders to the top in each layer. Then we will describe how to propagate an assignment to the last layer through the CLA to finally reach an assignment to the output bits.	

\begin{center}
\begin{minipage}[c]{0.9\textwidth}
\begin{algorithm}[H]
 \caption{Propagates from the initial layer $\ell=0$ to the final layer $\ell=h$ of the critical strip $L$ while assigning at most a constant number of bits per subcolumn.}
\begin{algorithmic}[1]
\label{wallace_alg1}
 \For{ $j = 0,1,\ldots,\lceil n/3 \rceil $}
  \State Branch on the inputs to the $j$-th row of adders $\{A^{xy}_{0,i,j}\}_i$.
  \For{ {\label{algstep0}}each layer $\ell = 0,1,\ldots, h-1$ before the last layer}
	\If{layer $\ell$ has a fully assigned row of adders $\{A^{xy}_{\ell,i,j'}\}_i$}
   		\State Propagate to tableau rows $2j', 2j'+1$ of layer $\ell+1$. {\label{algstep1}}
   		\State Merge to forget the assignment to the row of adders $\{A^{xy}_{\ell,i,j'}\}_i$. 
        \State Branch on any input carry-bits in tableau rows $2j', 2j'+1$ of layer $\ell+1$. {\label{algstep2}}
    \EndIf
  \EndFor
{\label{algstep3}}
\EndFor
\end{algorithmic}
\end{algorithm}
\end{minipage}
\end{center}

The branching program $B$ begins by following the Algorithm~\ref{wallace_alg1} on circuit $L$. We use the propagation loop in lines~\ref{algstep0}-\ref{algstep3} for circuit $R$, leaving the branching steps to circuit $L$. We claim that at the end, $B$ will reach an assignment to just the last layer of circuits $L$ and $R$. This will follow immediately from Lemma~\ref{wallace_lemma}.

\begin{lemma}
\label{wallace_lemma}
During the execution of Algorithm~\ref{wallace_alg1}, the tableau variables within each layer of circuit $L$ get assigned in row order from top to bottom. Furthermore, each tableau variable eventually receives an assignment.

Likewise, the tableau variables in each layer $\ell > 0$ of circuit $R$ get assigned in row order from bottom to top, and each tableau variable eventually receives an assignment.
\end{lemma}

\begin{proof}
We prove both properties in circuit $L$ by induction, making use of the row-friendliness of Wallace tree multipliers from Lemma~\ref{thm:nice_adders}. It is clear that the initial layer satisfies both properties. Suppose that layer $\ell-1$ satisfies both properties. Then its rows of adders $\{A^{xy}_{\ell-1,i,j'}\}_i$ get assigned to in ascending order with $j'=0,1,\ldots$. For each increment of $j'$, by row friendliness the steps~\ref{algstep1} and~\ref{algstep2} yield an assignment to all the variables in tableau rows $2j', 2j'+1$ of layer $\ell$. So layer $\ell$ gets assigned in row order from top to bottom, and each tableau variable in $\ell$ eventually receives an assignment.

The proof for circuit $R$ is symmetric, except the initial tableau is not assigned in horizontal rows, but rather diagonal rows. Nevertheless, the subsequent layer $\ell=1$ will still satisfy both desired properties and the induction argument may be used from there.
\end{proof}

\begin{corollary}
At the end of Algorithm~\ref{wallace_alg1}, the branching program $B$ reaches an assignment to precisely both rows in the last layer of circuits $L$ and $R$.
\end{corollary}

To propagate an assignment to last layer of $L$ or $R$ through the CLA, we will follow Algorithm~\ref{CLA_alg}. This algorithm will essentially perform a a post-order traversal of the full CLA tree. While it is not technically necessary to include the components of the CLA to the right of the critical strip, we have retained them for clarity.

\begin{center}
\begin{minipage}[c]{0.9\textwidth}
\begin{algorithm}[H]
 \caption{Propagates from the inputs to the critical strip outputs of the CLA while assigning at most a constant number of bits per CLA layer.}
\begin{algorithmic}[1]
\label{CLA_alg}
 \For{ $i = 0,1,\ldots, 2n $}
 \State Branch on any unassigned inputs to the $i$-th column: $t_{h,i,0}, t_{h,i,1}$.
 \While{ there is a pair of propagate and generate variables $p_{\ell,i'},g_{\ell,i'}$ with all their input variables assigned.}
 \State Propagate to $p_{\ell, i'},g_{\ell,i'}$ while merging to forget their input bits.
  \State Merge to forget the carry-bits computed by the CLA that output $p_{\ell, i'},g_{\ell,i'}$.
  	 \State Propagate to each carry-bit with all its input variables assigned.
     \State Propagate to each critical strip output bit with all its inputs assigned.
  
  \EndWhile
 \EndFor
\end{algorithmic}
\end{algorithm}
\end{minipage}
\end{center}

\begin{figure}[t]
  \centering 
\includegraphics[width=0.9\textwidth]{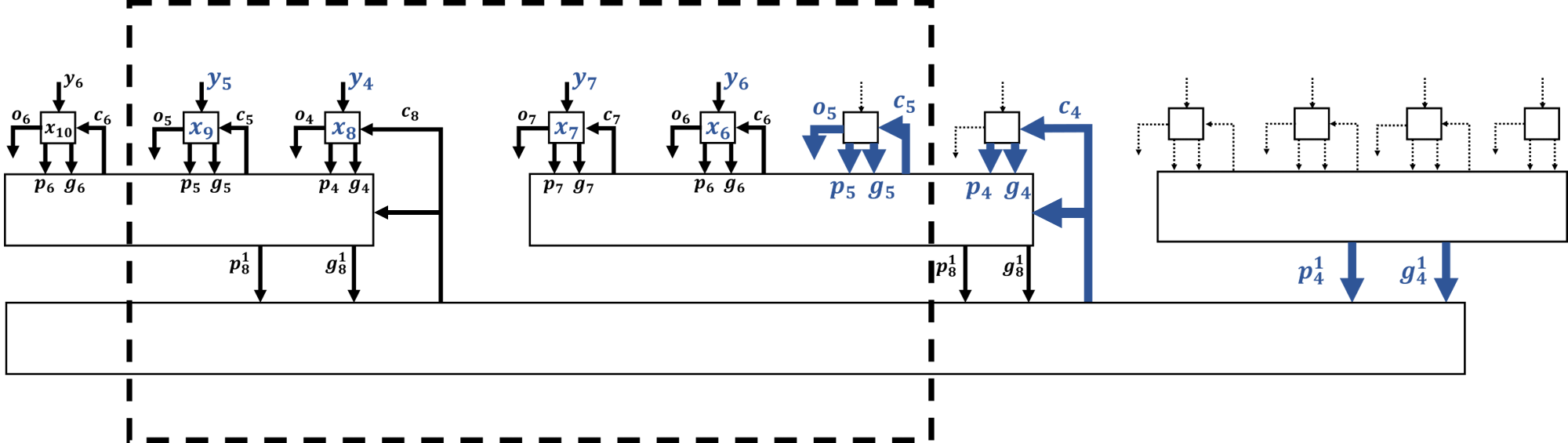}
\caption{An intermediate state in the CLA after scanning up to the sixth column. The box contains the columns of the critical strip. The blue variables are assigned while the blank variables were previously assigned, but then erased. Notice that we remember the assignment to the output variables in the strip and forgot the assignment outside.}
\label{cla_strip}
\end{figure}

After running Algorithm~\ref{CLA_alg} in both circuits $L$ and $R$, we have an assignment to the outputs of both critical strips. By Lemma~\ref{wallace_unsat}, this assignment violates an inequality-constraint in $E$.

\paragraph{Size Bound:}
We claim that in the first phase, where the branching program $B$ is executing Algorithm~\ref{wallace_alg1}, each node in $B$ is labeled by an assignment to at most four rows of tableau variables within each layer $\ell$ of $L$, and likewise for each layer $\ell>1$ for $R$. By Lemma~\ref{wallace_lemma}, the tableau variables within each layer are assigned in row order from top to bottom in $L$. So if four rows are assigned in a layer $\ell$, they form a fully assigned row of adders $\{A^{xy}_{0,i,j}\}_i$. Algorithm~\ref{wallace_alg1} will propagate that assignment to the next layer, erasing the assignment to the row of adders $\{A^{xy}_{0,i,j}\}_i$. The same proof works to show that at most four rows of tableau variables are assigned within each layer $\ell >1$ of $R$.

Each node in the first phase of $B$ then holds an assignment to at most $8 \delta h$ variables of the critical strip. Both $L$ and $R$ have at most $2n$ rows of tableau variables, so the number of tableau variables in the critical strip is upper bounded by $4nh$. Therefore the execution of Algorithm~\ref{wallace_alg1} will take at most $4nh$ steps. As this algorithm is also oblivious, each node gets labeled by an assignment to one of $4nh$ sets of at most $8 \delta h$ tableau variables. So the total number of nodes in the first phase of $B$ is at most $4nh 2^{\delta h} = 2^{16 \log ^2 n + O(\log n)}$.

We can obtain a more efficient version of Algorithm~\ref{wallace_alg1} by immediately propagating when an individual adder becomes fully assigned. This modified algorithm will only store at most two variables per subcolumn, except for a single "working" subcolumn in each layer that may hold three variables. This modification results in a size bound of $2^{8 \log^2 n + O(\log n)}$.

We give a polynomial bound for the second phase, where the branching program $B$ is executing Algorithm~\ref{CLA_alg}. Observe that this algorithm only keeps an assignment to variables within the sub-CLAs intersecting the $i$-th column. At most one sub-CLA in each of the $\log_4 n$ layers will intersect the $i$-th column, so there are $O(\log n)$ assigned variables in any step of Algorithm~\ref{CLA_alg}. The whole CLA has $O(n)$ variables, therefore $B$ uses a polynomial number of nodes to execute Algorithm~\ref{CLA_alg}.

The total size of the branching program $B$ is then $2^{8 \log^2 n + O(\log n)}$.
\end{proof}

\begin{theorem}
There is a regular resolution proof of size $2^{8 \log^2 n + O(\log n)}$ that $\phi^{\mathrm {Wall}}_{\mathrm{Comm}}(n)$ is unsatisfiable
\end{theorem}

\begin{proof}
As usual, we initially branch on the assignments $\sigma_e(k) = \{e_0=0,e_1=0,\ldots e_k=1\}$ for $k \in [0,2n-1]$. The $k$-th branch contains the clauses $\strip(k)$ so we can use the Read-Once branching program from Lemma~\ref{thm:walStrip} (with each node augmented with the assignment $\sigma_e(k)$) to show that the branch is unsatisfiable.
\end{proof}

\subsection{Proofs of Wallace Tree Multiplier Distributivity}

Our proof of commutativity for Wallace tree multipliers used Algorithms~\ref{wallace_alg1} and \ref{CLA_alg} to efficiently propagate an assignment from the initial layer of $L$'s critical strip to the outputs. We will modify the branching step in these algorithms to verify the distributivity of Wallace tree multipliers.

\begin{definition}
Define a SAT instance $\phi^{\mathrm {Wall}}_{\mathrm{Dist}}(n)$ encoding the identity $x(y+z) = xy+xz$ in the usual way, with subcircuits $L^{y+z},L^{x(y+z)}$ forming circuit $L$, $R^{xy},R^{xz},R^{xy+xz}$ forming circuit $R$ and inequality-constraints $E$.
\end{definition}

\begin{theorem}
There is a regular resolution proof of size $2^{O(\log^2 n)}$ that $\phi^{\mathrm {Wall}}_{\mathrm{Dist}}(n)$ is unsatisfiable
\end{theorem}

\begin{proof}
(Sketch) We sketch the proofs for distributivity as they are simpler than the proofs for commutativity. The main difference is that we branch on the input variables $\mathbf{x,y,z}$ rather than the tableau variables in the initial layer.

We define critical strips in the usual way for each multiplier. There are at most $n+2$ unconstrained carry bits in the $n+1$-bit multiplier $L^{x(y+z)}$ and one unconstrained carry bit from the adder $L^{y+z}$ for $n+3$ total in $L$'s critical strip. Together, the two $n$-bit multipliers $R^{xy}, R^{xz}$ have $2n+2$ unconstrained carry bits. The adder $R^{xy+xz}$ contributes one more for a total of $2n+3$ unconstrained carry bits in $R$'s critical strip. So if our critical strip has width $\delta = \log (2n+4)$, it will be unsatisfiable.

We now describe a branching program $B$ that proves a given critical strip $\phi_{\textrm{Strip}}(k)$ is unsatisfiable. We begin the branching program $B$ by running Algorithm~\ref{wallace_alg1} with the following modification: instead of branching on a row of initial tableau variables in some multiplier $\{ t_{0,i,j} \}_i$, branching program $B$ will instead branch on the input variables $\mathbf{x,y,z}$ and propagate to that row of tableau variables $\{ t_{0,i,j} \}_i$. To reveal the rows from top to bottom in the initial layer of each multiplier's critical strip, we only need to assign a sliding window of $\delta$ bits in each input bitvector $\mathbf{x,y,z}$. The resulting branch order on $\mathbf{x,y,z}$ is the same as in our proof of array multiplier distributivity.

At the end of Algorithm~\ref{wallace_alg1}, the branching program $B$ reaches an assignment to the last layer of each multiplier $R^{xy}, R^{xz}, L^{x(y+z)}$. By using Algorithm~\ref{CLA_alg}, we propagate this assignment to the multiplier outputs $\mathbf{xy,xz}$ and $\mathbf{x(y+z)}$. Lastly, we propagate from $\mathbf{xy,xz}$, through the CLA circuit $L^{xy+xz}$, to the final output $\mathbf{xy+xz}$. Since the critical strip was unsatisfiable, the resulting assignment to $\mathbf{x(y+z)}$ and $\mathbf{xy+xz}$ must violate some equality-constraint from $E$.
\end{proof}

\subsection{Degree Two Identity Proofs for Wallace Tree Multipliers}

Using the same ordering on the input variables and ideas from the proof of Theorem~\ref{thm:deg2array}, we can prove the analogous result for Wallace tree multipliers.

\begin{theorem}
For any degree two ring identity $L=R$, there are quasipolynomial size regular refutations for $\phi^{\mathrm{Wall}}_{L=R}(n)$.
\end{theorem}

\update{

\section{Proving Equivalence Between Multipliers}

Given any two $n$-bit multiplier circuits $\otimes_{1}$
and $\otimes_{2}$ we can define a Boolean formula $\phi_{\otimes_{1}=\otimes_{2}}$ encoding the negation of the identity
$\mathbf{x} \otimes_{1} \mathbf{y} = \mathbf{x} \otimes_{2} \mathbf{y}$ between length $n$ bitvectors $\mathbf{x}$ and $\mathbf{y}$.

If both $\otimes_{1}$ and $\otimes_{2}$ are correct and compute using the typical tableau for 
multipliers then, as before, we can split $\phi_{\otimes_{1}=\otimes_{2}}$
into unsatisfiable critical strips. 
We can scan down both strips row-by-row, as in the proofs for commutativity and distributivity. If we have reached the outputs of both multipliers without finding an error, these outputs will disagree with the inequality-constraints for the critical strip.
For our examples this method yields polynomial-size proofs if neither is a
Wallace tree multiplier, and quasi-polynomial size proofs otherwise.

On the other hand, if one multiplier is incorrect and the other is not, then
the proof search will yield a satisfying assignment in the appropriate
critical strip.

In the more general case where a multiplier does not use
the typical tableau, one can label each internal gate by the index of the
smallest output bit to which it is connected and focus on comparing subcircuits
labeled by $O(\log n)$ consecutive output bits, as we do with critical
strips.  The complexity of this equivalence checking will depend somewhat on
the similarity of the circuits involved.
}

\section{Discussion}

% Original
Despite significant advances in SAT solvers, one of their key persisting
weaknesses has been in verifying arithmetic circuits containing multipliers.
This pointed towards the conjecture that that the corresponding resolution proofs
are exponentially large; if true, this would have been a fundamental obstacle
putting nonlinear arithmetic out of reach for any CDCL SAT solver.

Thus, much of the recent research on multiplier verification has focused on
using algebraic reasoning, in particular Groebner basis methods. 
%In this approach, checking the validity of a multiplier reduces to solving an 
%instance of the \emph{ideal membership problem}. \citeN{DBLP:conf/date/Sayed-AhmedGKSD16} showed that Groebner basis methods could be used to verify several different types of multiplier design.
The recent work of Ritirc, Biere, and Kauers~\cite{bkr17,DBLP:conf/date/RitircBK18} has improved the Groebner basis approach by dividing
 a multiplier into columns, and then incrementally checking that each column receives
 and transmits its carry-bits correctly. They find that this incremental
  method allows off-the-shelf computer algebra software to verify
  "simple" multiplier designs of up to 64 bits, though "optimized" multipliers still pose some difficulty.

We have shown that the conjectured resolution proof size barrier does not
 hold by giving the first small
resolution proofs for verifying any degree two ring identity for the most
common multiplier designs. 
We introduced a method of dividing each instance into narrow, but still
unsatisfiable, critical strips that is sufficiently general to yield 
short proofs for a wide variety of popular multiplier designs. 
In light of our results and~\cite{bkr17,DBLP:conf/date/RitircBK18}, it seems that for verifying multipliers at
 the bit-level, the column-wise view is most natural.
 This is in contrast to the row-wise view taken, for example, 
 in verifying multipliers at the word level.
We remark that the critical strip decomposition is not only useful
 in the domain of resolution proofs. Other verification methods may 
 find critical strips a useful testing ground, or could even benefit from 
 checking each strip instead of the full multiplier all at once. 

 Given the historical success of CDCL SAT solvers for
finding specific proofs, our results suggest a new path towards verifying
nonlinear arithmetic. The proof size
upper bounds we derived were conservative; we did not try to optimize the
parameters. Nevertheless, the observed scaling of SAT solver performance
 on these problems suggests that they do not currently find proofs matching
  even these upper bounds. An important direction for improving SAT solvers 
  is to find the right guiding
information to add, either to the formulas derived from the circuits or to
CDCL SAT solver heuristics, to help them find shorter proofs.

It also remains open to find a small resolution proof verifying the last ring
property, associativity $(xy)z =x(yz)$. 
Our critical strip idea alone does not seem to work:  while we can divide the
outer multipliers into narrow critical strips, the $yz$ or $xy$ multipliers
remain intact. 
These critical strips do not seem to have small cuts.
Finding efficient proofs of associativity, combined with our results for degree
two identities, could yield small proofs of any general ring identity.
 
\bibliographystyle{plain}
\bibliography{multiply} 
%\newpage
%\appendix
%\input{comments}
\end{document}